\documentclass{article}
\newcommand{\PaperPreprint}{}

\usepackage{iclr2027_conference,times}
\usepackage[utf8]{inputenc} 
\usepackage[T1]{fontenc}    
\ifdefined\XeTeXversion
  \usepackage{newunicodechar}
  \newunicodechar{–}{--}
\fi
\usepackage{hyperref}       
\makeatletter
\renewcommand*{\hyper@natlinkbreak}[2]{#1}
\makeatother
\usepackage{url}            
\usepackage{booktabs}       
\usepackage{amsfonts}       
\usepackage{nicefrac}       
\usepackage{microtype}      
\usepackage{xcolor}         
\usepackage{wrapfig}

\usepackage{caption}
\usepackage{subcaption}
\usepackage{amsmath,amssymb}
\usepackage{xcolor}
\usepackage{graphicx}
\usepackage{amsmath,amssymb}
\usepackage{xcolor}
\usepackage{graphicx}
\usepackage{amsmath,amssymb}
\usepackage{xcolor}
\usepackage{graphicx}
\usepackage{tikz}
\usetikzlibrary{calc,arrows.meta,positioning,shadows.blur}

\definecolor{PCPurple}{RGB}{92,67,156}
\definecolor{PCPurpleLight}{RGB}{247,244,253}
\definecolor{PCBlue}{RGB}{59,116,200}
\definecolor{PCBlueLight}{RGB}{241,247,255}
\definecolor{PCRed}{RGB}{211,82,68}
\definecolor{PCRedLight}{RGB}{252,240,238}
\definecolor{PCGreen}{RGB}{52,150,95}
\definecolor{PCGreenLight}{RGB}{240,250,244}
\definecolor{PCGold}{RGB}{196,139,26}
\definecolor{SoftGray}{RGB}{248,248,250}
\definecolor{Ink}{RGB}{45,45,52}

\tikzset{
  outerpanel/.style={
    rounded corners=5pt,
    draw=black!14,
    fill=white,
    line width=.45pt,
    blur shadow={
      shadow blur steps=4,
      shadow xshift=.7pt,
      shadow yshift=-.7pt,
      shadow opacity=9
    }
  },
  scenebox/.style={
    rounded corners=4pt,
    draw=black!10,
    fill=SoftGray,
    line width=.35pt
  },
  ax/.style={
    -{Latex[length=2.2mm,width=1.55mm]},
    line width=.72pt,
    draw=black!70
  },
  grid/.style={
    draw=black!18,
    line width=.30pt
  },
  floor/.style={
    fill=black!02,
    draw=black!22,
    line width=.38pt
  },
  priorline/.style={
    draw=PCPurple!55,
    dashed,
    line width=.55pt
  },
  projection/.style={
    draw=black!24,
    dashed,
    line width=.32pt
  },
  icouple/.style={
    draw=PCBlue!85!black,
    line width=.82pt
  },
  kcouple/.style={
    draw=PCRed!85!black,
    line width=.82pt
  },
  fullcouple/.style={
    draw=PCGold!88!black,
    line width=1.02pt
  },
  topdot/.style={
    circle,
    draw=white,
    line width=.25pt,
    inner sep=1.45pt
  },
  floordot/.style={
    circle,
    draw=white,
    line width=.22pt,
    inner sep=1.05pt
  },
  smallbox/.style={
    rounded corners=2pt,
    draw=black!12,
    fill=white,
    line width=.32pt,
    inner sep=2pt
  },
  note/.style={
    font=\tiny,
    text=Ink!68
  }
}
\usepackage{amsmath}
\usepackage{amssymb}
\usepackage{mathtools}
\usepackage{amsthm}

\usepackage{algorithm}
\usepackage{algpseudocode} 
\usepackage{multirow} 
\usepackage{svg}

\theoremstyle{plain}
\newtheorem{theorem}{Theorem}[section]

\newtheorem{lemma}[theorem]{Lemma}

\theoremstyle{definition}

\newtheorem{assumption}[theorem]{Assumption}
\theoremstyle{remark}

\usepackage{wrapfig}
\usepackage{caption}
\title{Learning Coordinated Preference for Multi-Objective Multi-Agent Reinforcement Learning}

\author{%
Pengxin Wang, Lihao Guo, Yi Xie\\
\textbf{Siyang Cao, Bo Liu, Jingdi Chen}\\
Department of Electrical and Computer Engineering\\
University of Arizona
}

\ifdefined\PaperPreprint
  \iclrfinalcopy
  \let\PaperOriginalMakeTitle\maketitle
  \renewcommand{\maketitle}{%
    \PaperOriginalMakeTitle
    \lhead{Preprint}%
    \ifdefined\AuthorMetadataPending
      \begin{center}
        \small\bfseries Draft preprint: author metadata is incomplete.\\
        Complete authors.tex before distributing this version.
      \end{center}
    \fi
  }
\else
  \iclrfinalfalse
\fi

\begin{document}

\maketitle

\begin{abstract}
Cooperative multi-objective multi-agent reinforcement learning (MOMARL) models team decision making under multiple, potentially conflicting objectives. In this setting, conflicts arise not only across objectives but also across agents with different observations, roles, and contributions. We propose Preference Coordinated Multi-agent Policy Optimization (PCMA), which learns coordinated agent-specific preferences to enable complementary trade-offs among agents. Theoretically, we formulate cooperative MOMARL as a team-optimal equilibrium problem, and show that, under suitable conditions, preference diversity yields a first-order improvement in the team objective. Experiments on multiple cooperative MOMA environments and a practical traffic-control scenario show that PCMA improves both performance and trade-off coordination.
\end{abstract}
\section{Introduction}
\label{sec:intro}

Multi-objective Multi-agent reinforcement learning (MOMARL) provides a natural framework for modeling complex decision making systems where multiple agents must coordinate under competing objectives. Existing applications on domains such as traffic signal control, dynamic traffic management, and natural resource allocation show that such systems often require agents to balance competing objectives such as efficiency, safety, energy consumption, and fairness \citep{MOMA_real1, MOMA_real2, MOMA_real3}. Therefore, modeling such problems with a single scalar reward is often inadequate \citep{ScalarRewardNotEnough}. At the same time, conflict structure is extremely intricate in MOMA problem: conflicts may arise not only between objectives within each agent, but also between agents under the same objective and across agents with different objective priorities. Therefore, MOMARL is both practically important and theoretically challenging, calling for learning algorithms that model and coordinate such intertwined objectives in Multi-agent systems.

Despite growing interest, research on MOMARL remains at an early stage. Several notable efforts have begun to structure the field, including the utility based analysis of Multi-objective Multi-agent decision making~\citep{MOMADecisionMaking} and the introduction of MOMALand~\citep{momaland}, as the first benchmark environments in this field. A straight-forward approach adopted by existing works is to optimize scalarized objective using the same preference vector across all agents. However, enforcing identical objective emphasis for every agent can limit the effectiveness of coordination. When all agents prioritize objectives in the same way, they may compete along the same dimensions of the reward. For example, consider two autonomous vehicles passing through an unsignalized intersection. If both vehicles strongly prefer efficiency, they may aggressively enter the intersection and cause a collision. If both vehicles strongly prefer safety, they may both wait and waste time. A better joint behavior requires coordinated preferences: one vehicle may emphasize yielding and safety, while the other emphasizes passing efficiency.

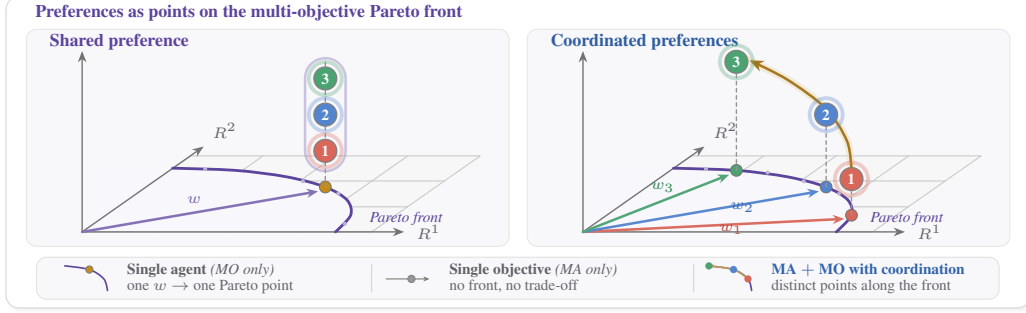
\begin{figure*}[t]
\centering
\resizebox{0.98\textwidth}{!}{%
\begin{tikzpicture}[font=\scriptsize, line cap=round, line join=round,
    paretofront/.style={draw=PCPurple, line width=1.1pt, line cap=round},
    paretofeasible/.style={fill=PCPurple!4, draw=none},
    paretocontinuum/.style={fill=PCPurple!40, draw=none},
    paretodot/.style={circle, inner sep=0pt, minimum size=4.6pt,
                       draw=black!55, line width=.35pt},
    agenthalo/.style={circle, line width=1.5pt, opacity=.32,
                       inner sep=0pt, minimum size=13pt},
    agentdot/.style={circle, draw=black!50, line width=.5pt,
                       inner sep=.4pt, minimum size=9pt,
                       font=\tiny\bfseries, text=white},
    sharedring/.style={draw=PCPurple!55, line width=1pt, opacity=.55,
                       rounded corners=.28cm},
    droppole/.style={draw=black!42, line width=.55pt,
                       dash pattern=on 1.4pt off 1.4pt},
    prefarrow/.style={-{Stealth[length=4.5pt, width=3.6pt]}, line width=1pt},
    floorline/.style={draw=black!22, line width=.4pt},
    axline/.style={-{Stealth[length=4.5pt, width=3.4pt]},
                       draw=black!55, line width=.55pt},
    miniarrow/.style={->, >={Stealth[length=2.5pt, width=2pt]},
                       draw=black!50, line width=.45pt},
]

\path[use as bounding box] (-.05,-.08) rectangle (14.25,4.42);
\draw[outerpanel] (0,0) rectangle (14.15,4.30);

\node[anchor=west, font=\scriptsize\bfseries, text=PCPurple]
  at (.28,4.08) {Preferences as points on the multi-objective Pareto front};

\draw[scenebox] (.28,.85) rectangle (6.95,3.85);

\node[anchor=base west, font=\scriptsize\bfseries, text=PCPurple]
  at (.48,3.62) {Shared preference};

\begin{scope}[shift={(1.05,1.05)}]
  \fill[Ink!4]      (0,0) -- (4,0) -- (5.5,1.05) -- (1.5,1.05) -- cycle;
  \draw[floorline]  (0,0) -- (4,0) -- (5.5,1.05) -- (1.5,1.05) -- cycle;
  \foreach \t in {1,2,3}{
      \draw[floorline] (\t,0) -- ($(\t,0)+(1.5,1.05)$);
  }
  \foreach \t in {.33,.67}{
      \draw[floorline] ($(0,0)+\t*(1.5,1.05)$) -- ($(4,0)+\t*(1.5,1.05)$);
  }
  \fill[paretofeasible]
      (0,0) -- (3.5,0) .. controls (3.78,.10) and (3.80,.18) ..
      (3.71,.23) .. controls (3.74,.34) and (3.72,.40) ..
      (3.66,.44) .. controls (3.55,.54) and (3.45,.59) ..
      (3.36,.62) .. controls (3.10,.71) and (2.95,.74) ..
      (2.83,.76) .. controls (2.55,.82) and (2.30,.85) ..
      (2.12,.85) .. controls (1.78,.87) and (1.50,.88) ..
      (1.25,.88) -- (0,.88) -- cycle;

  \draw[axline] (0,0) -- (4.45,0);
  \draw[axline] (0,0) -- (1.68,1.18);
  \draw[axline] (0,0) -- (0,2.50); 
  \node[font=\tiny, text=Ink!75, anchor=west]        at (4.50,0)    {$R^1$};
  \node[font=\tiny, text=Ink!75, anchor=south west] at (1.68,1.18) {$R^2$};

  \draw[paretofront] plot[smooth, tension=.55] coordinates
      {(3.5,0) (3.71,.23) (3.66,.44) (3.36,.62) (2.83,.76) (2.12,.85) (1.25,.88)};
  \foreach \x/\y in {3.63/.11, 3.54/.53, 3.05/.70, 1.70/.87}{
      \fill[paretocontinuum] (\x,\y) circle[radius=.028];
  }
  \node[font=\tiny\itshape, text=PCPurple, anchor=west]
      at (3.85,.18) {Pareto front};

  \coordinate (ach) at (3.36,.62);
  \draw[prefarrow, draw=PCPurple!75] (0,0) -- ($(ach)+(-.10,-.06)$);
  \node[font=\tiny, text=PCPurple!85] at (1.55,.45) {$w$};
  \node[paretodot, fill=PCGold] at (ach) {};

  \def\altone{.50}\def\alttwo{1.00}\def\altthree{1.50}
  \draw[droppole] (ach) -- ($(ach)+(0,1.75)$);
  \draw[sharedring]
      ($(ach)+(-.28,\altone-.28)$) rectangle ($(ach)+(.28,\altthree+.28)$);

  \node[agenthalo, draw=PCRed!85]   at ($(ach)+(0,\altone)$)   {};
  \node[agentdot,  fill=PCRed!88]   at ($(ach)+(0,\altone)$)   {1};
  \node[agenthalo, draw=PCBlue!85]  at ($(ach)+(0,\alttwo)$)   {};
  \node[agentdot,  fill=PCBlue!88]  at ($(ach)+(0,\alttwo)$)   {2};
  \node[agenthalo, draw=PCGreen!85] at ($(ach)+(0,\altthree)$) {};
  \node[agentdot,  fill=PCGreen!88] at ($(ach)+(0,\altthree)$) {3};
\end{scope}

\draw[scenebox] (7.20,.85) rectangle (13.87,3.85);

\node[anchor=base west, font=\scriptsize\bfseries, text=PCBlue!80!black]
  at (7.40,3.62) {Coordinated preferences};

\begin{scope}[shift={(7.97,1.05)}]
  \fill[Ink!4]      (0,0) -- (4,0) -- (5.5,1.05) -- (1.5,1.05) -- cycle;
  \draw[floorline]  (0,0) -- (4,0) -- (5.5,1.05) -- (1.5,1.05) -- cycle;
  \foreach \t in {1,2,3}{
      \draw[floorline] (\t,0) -- ($(\t,0)+(1.5,1.05)$);
  }
  \foreach \t in {.33,.67}{
      \draw[floorline] ($(0,0)+\t*(1.5,1.05)$) -- ($(4,0)+\t*(1.5,1.05)$);
  }
  \fill[paretofeasible]
      (0,0) -- (3.5,0) .. controls (3.78,.10) and (3.80,.18) ..
      (3.71,.23) .. controls (3.74,.34) and (3.72,.40) ..
      (3.66,.44) .. controls (3.55,.54) and (3.45,.59) ..
      (3.36,.62) .. controls (3.10,.71) and (2.95,.74) ..
      (2.83,.76) .. controls (2.55,.82) and (2.30,.85) ..
      (2.12,.85) .. controls (1.78,.87) and (1.50,.88) ..
      (1.25,.88) -- (0,.88) -- cycle;

  \draw[axline] (0,0) -- (4.45,0);
  \draw[axline] (0,0) -- (1.68,1.18);
  \draw[axline] (0,0) -- (0,2.50); 
  \node[font=\tiny, text=Ink!75, anchor=west]        at (4.50,0)    {$R^1$};
  \node[font=\tiny, text=Ink!75, anchor=south west] at (1.68,1.18) {$R^2$};

  \draw[paretofront] plot[smooth, tension=.55] coordinates
      {(3.5,0) (3.71,.23) (3.66,.44) (3.36,.62) (2.83,.76) (2.12,.85) (1.25,.88)};
  \foreach \x/\y in {3.54/.53, 3.05/.70, 2.55/.82, 1.70/.87}{
      \fill[paretocontinuum] (\x,\y) circle[radius=.028];
  }
  \node[font=\tiny\itshape, text=PCPurple, anchor=west]
      at (3.85,.18) {Pareto front};

  \coordinate (pA) at (3.71,.23);
  \coordinate (pB) at (3.36,.62);
  \coordinate (pC) at (2.12,.85);

  \draw[prefarrow, draw=PCRed!85]   (0,0) -- ($(pA)+(-.10,-.05)$);
  \draw[prefarrow, draw=PCBlue!85]  (0,0) -- ($(pB)+(-.10,-.06)$);
  \draw[prefarrow, draw=PCGreen!85] (0,0) -- ($(pC)+(-.10,-.06)$);
  \node[font=\tiny, text=PCRed!88!black]   at (2.05,.04) {$w_1$};
  \node[font=\tiny, text=PCBlue!88!black]  at (2.20,.34) {$w_2$};
  \node[font=\tiny, text=PCGreen!88!black] at (1.10,.62) {$w_3$};

  \node[paretodot, fill=PCRed!88]   at (pA) {};
  \node[paretodot, fill=PCBlue!88]  at (pB) {};
  \node[paretodot, fill=PCGreen!88] at (pC) {};

  \def\altone{.50}\def\alttwo{1.00}\def\altthree{1.50}
  \draw[droppole] (pA) -- ($(pA)+(0,\altone+.10)$);
  \draw[droppole] (pB) -- ($(pB)+(0,\alttwo+.10)$);
  \draw[droppole] (pC) -- ($(pC)+(0,\altthree+.10)$);

  \draw[line width=3pt, draw=PCGold!28, opacity=.55, line cap=round]
      plot[smooth, tension=.55] coordinates
      {(3.71,0.85) (3.65,1.25) (3.36,1.70) (2.83,2.10) (2.12,2.45)};
  \draw[line width=1pt, draw=PCGold!82!black, opacity=.95, line cap=round,
        -{Stealth[length=5.5pt, width=4.5pt]}]
      plot[smooth, tension=.55] coordinates
      {(3.71,0.85) (3.65,1.25) (3.36,1.70) (2.83,2.10) (2.30,2.37)};

  \node[agenthalo, draw=PCRed!85]   at ($(pA)+(0,\altone)$)   {};
  \node[agentdot,  fill=PCRed!88]   at ($(pA)+(0,\altone)$)   {1};
  \node[agenthalo, draw=PCBlue!85]  at ($(pB)+(0,\alttwo)$)   {};
  \node[agentdot,  fill=PCBlue!88]  at ($(pB)+(0,\alttwo)$)   {2};
  \node[agenthalo, draw=PCGreen!85] at ($(pC)+(0,\altthree)$) {};
  \node[agentdot,  fill=PCGreen!88] at ($(pC)+(0,\altthree)$) {3};
\end{scope}

\draw[rounded corners=4pt, draw=black!12, fill=SoftGray, line width=.36pt]
  (.42,.12) rectangle (13.72,.70);

\begin{scope}[shift={(.85,.40)}]
    \draw[draw=PCPurple, line width=.65pt, line cap=round]
        plot[smooth, tension=.55] coordinates
        {(.55,-.16) (.50,0) (.30,.13) (-.04,.18)};
    \fill[PCGold]                              (.30,.13) circle[radius=.05];
    \draw[draw=black!50, line width=.3pt]      (.30,.13) circle[radius=.05];
\end{scope}
\node[anchor=west, font=\tiny\bfseries, text=Ink!75] at (1.55,.52)
    {Single agent {\normalfont\itshape (MO only)}};
\node[anchor=west, font=\tiny, text=Ink!72] at (1.55,.30)
    {one $w$ $\to$ one Pareto point};

\draw[draw=black!18, line width=.3pt] (5.05,.22) -- (5.05,.60);

\begin{scope}[shift={(5.30,.40)}]
    \draw[miniarrow] (-.05,0) -- (.55,0);
    \fill[Ink!50]                              (.30,0)   circle[radius=.05];
    \draw[draw=black!50, line width=.3pt]      (.30,0)   circle[radius=.05];
\end{scope}
\node[anchor=west, font=\tiny\bfseries, text=Ink!75] at (6.00,.52)
    {Single objective {\normalfont\itshape (MA only)}};
\node[anchor=west, font=\tiny, text=Ink!72] at (6.00,.30)
    {no front, no trade-off};

\draw[draw=black!18, line width=.3pt] (9.50,.22) -- (9.50,.60);

\begin{scope}[shift={(9.75,.40)}]
    \draw[draw=PCPurple, line width=.65pt, line cap=round]
        plot[smooth, tension=.55] coordinates
        {(.55,-.16) (.50,0) (.30,.13) (-.04,.18)};
    \draw[draw=PCGold!75, line width=.7pt, opacity=.85, line cap=round,
          -{Stealth[length=2.5pt, width=2pt]}]
        plot[smooth, tension=.5] coordinates
        {(.50,0) (.30,.13) (-.06,.185)};
    \fill[PCRed!88]                            (.50,0)    circle[radius=.05];
    \fill[PCBlue!88]                           (.30,.13)  circle[radius=.05];
    \fill[PCGreen!88]                          (-.04,.18) circle[radius=.05];
\end{scope}
\node[anchor=west, font=\tiny\bfseries, text=PCBlue!85!black] at (10.45,.52)
    {MA $+$ MO with coordination};
\node[anchor=west, font=\tiny, text=Ink!72] at (10.45,.30)
    {distinct points along the front};

\end{tikzpicture}%
}
\caption{
Illustration of preference coordination. Agents' policies are projected on to the multi-objective return space $(R^1, R^2)$. A preference $w$ from the origin selects one point on the front. \textbf{Left:} all agents use the same preference vector $w$, which could cause conflicts. \textbf{Right:} distinct $w_1, w_2, w_3$ project to three different points along the front and the team jointly covers the front. A single agent can only occupy one point, a single objective produces no front at all, and only multi-agent $+$ multi-objective with coordination spreads agents along the front.}
\label{fig:illustration_intro}
\end{figure*}

This motivates our central idea: rather than forcing all agents to follow the same trade-off, we allow agents to adapt their preferences through a team-level coordination mechanism. From a multi-objective perspective~\citep{MORL_survey1, MORL_survey2}, each agent's policy can be viewed as inducing a point in a shared objective space. By conditioning policies on different preferences, agents can approach different regions of the Pareto front. Coordinating these preferences enables the multi-agent system to assign diverse but complementary roles, thereby reducing behavioral conflicts and improving team-level performance, as illustrated in  Figure.~\ref{fig:illustration_intro}.

To realize this idea, in this paper, we propose Preference Coordinated Multi-agent Policy Optimization (PCMA) under the centralized training with decentralized execution (CTDE) framework~\citep{MADDPG}. We model preference as a latent coordination variable and train a stochastic planner for each agent to adapt its preference based on team improvement and its local observation, and then each agent executes a policy conditioned on its sampled preference. To encourage meaningful specialization, we further regularize the learned preference distributions so that agents avoid collapsing to the same preference direction. Our contributions are as follows:
\begin{itemize}
    \item We formulate cooperative MOMARL as a multi-objective partially observable stochastic game (mo-POSG) with partially aligned interests, and adopt team-optimal equilibrium as the solution concept, seeking a preference profile and an induced equilibrium that maximize the shared team objective.

    \item Theoretically, we characterize first-order team improvement through a decomposition that isolates the contribution of preference diversity, and establish a local equilibrium-tracking result for preference-conditioned games under suitable conditions.

    \item We propose PCMA, a practical policy-based method with stochastic preference planners and preference-conditioned actors, and provide empirical evaluations on selected cooperative tasks from MPE, SMAC, and MOMAland, as well as a traffic-control validation scenario.
\end{itemize}

\section{Related Work}
\label{sec:related}
Current work on cooperative Multi-objective Multi-agent reinforcement learning (MOMARL) remains at an early stage of exploration. As a pioneer, work on Multi-objective games has put forward solution concepts such as Pareto Nash
equilibrium, which extends Nash equilibrium to vector valued utilities by requiring Pareto optimality under unilateral deviations~\citep{PNG_in_MOG1,PNG_in_MOG2,PNG_in_MOG4}. However, Pareto Nash equilibrium can be ill-posed as an optimization criterion for cooperative MOMARL. Since a policy profile is ruled out only if some unilateral deviation improves at least one objective without decreasing any other objective, many policies that are suboptimal under meaningful trade offs may still satisfy the Pareto Nash condition. We provide an illustrative example in Appendix~\ref{app:pareto-nash}, where all joint policies in the specified fixed-variance Gaussian class satisfy the Pareto Nash condition, while only a small subset corresponds to meaningful team level trade offs.

More recently, MOMALand~\citep{momaland} considers linear scalarization based approaches to convert the vector valued team reward into a scalar reward, and then train policies by standard MARL algorithms such as MAPPO \citep{MAPPO}. However, they need re-train a separate policy for new preference vector, making them impractical for desired trade off during execution. Another work, MoMix~\citep{momix}, extends the value decomposition principle of QMIX \citep{QMIX} to vector valued rewards. While this provides an important step toward Multi-objective credit assignment, it relies on strong individual global max (IGM) assumptions and can't be applied to continous action space.

\section{Preliminaries and Problem Formulation}
\label{sec:problem}

We introduce the problem setting in Subsection~\ref{sec:cooperative_momarl} and formulate the preference coordination objective in Subsection~\ref{sec:preference_coordination_objective}.

\subsection{Multi-objective Partially Observable Stochastic Game}
\label{sec:cooperative_momarl}

We consider a multi-objective partially observable stochastic game (mo-POSG) with partially aligned interests, following the formulation of \citet[Definition~1]{momaland}. Our setting involves \(N\) agents and \(K\) objectives and is specified by
\[
\mathcal{G}
=
\left\langle
N,
\mathcal{S},
\{\mathcal{A}_i\}_{i=1}^{N},
P,
\Omega,
O,
\{\mathbf R_i\}_{i=1}^{N},
\gamma
\right\rangle .
\]
Here, \(\mathcal{S}\) is the state space, \(\mathcal{A}=\prod_{i=1}^{N}\mathcal{A}_i\) is the joint action space, \(\Omega=\prod_{i=1}^{N}\Omega_i\) is the joint observation space, and \(\gamma\) is the discount factor. The environment evolves according to \(s_{t+1}\sim P(\cdot\mid s_t,\mathbf a_t)\), where \(\mathbf a_t=(a_{1,t},\ldots,a_{N,t})\). Each agent receives only its local observation \(o_i\). We use a state-dependent observation kernel \(O(\mathbf o\mid s)\), a special case of the general mo-POSG observation model. Each agent has a vector reward function \(\mathbf R_i:\mathcal{S}\times\mathcal{A}\times\mathcal{S}\rightarrow\mathbb R^K\).

All agents share the same team objective, such as winning a game, represented by a common scalar reward \(r_{\mathrm{team}}\). Agent-specific guiding vector rewards \(\mathbf r_i\) and preferences may introduce conflicting incentives, resulting in partially aligned interests. We instantiate the per-agent vector rewards as
\[
\mathbf R_i(s,\mathbf a,s')
=
r_{\mathrm{team}}(s,\mathbf a)\mathbf 1_K
+
\mathbf r_i(s,\mathbf a),
\]
where \(\mathbf 1_K\) is the all-ones vector, \(r_{\mathrm{team}}(s,\mathbf a)\in\mathbb R\) is a sparse team reward, and \(\mathbf r_i(s,\mathbf a)\in\mathbb R^K\) provides auxiliary feedback on objectives such as efficiency, cost, or risk.

\subsection{MOMARL as Team Optimal Equilibrium.}
\label{sec:preference_coordination_objective}
Let \(\theta_i\) parameterize agent \(i\)'s policy, and let \(\pi_{\mathbf\theta}\) denote the joint policy with \(\mathbf\theta=(\theta_1,\ldots,\theta_N)\). We define the shared team objective and the agent-specific vector objective as
\[
J_{\mathrm{team}}(\mathbf\theta)
:=
\mathbb E_{\pi_{\mathbf\theta}}
\left[
\sum_{t=0}^{T}\gamma^t r_{\mathrm{team}}(s_t,\mathbf a_t)
\right],
\qquad
\mathbf J_i(\mathbf\theta)
:=
\mathbb E_{\pi_{\mathbf\theta}}
\left[
\sum_{t=0}^{T}\gamma^t \mathbf r_i(s_t,\mathbf a_t)
\right].
\]
Given a preference profile \(\mathbf p=(p_1,\ldots,p_N)\), with \(p_i\in\Delta^{K-1}:=\{p\in\mathbb R_{\geq 0}^{K}:\mathbf 1_K^\top p=1\}\),
agent \(i\) evaluates its vector return using linear scalarization:
\[
\begin{aligned}
U_i(\mathbf\theta;p_i)
&:=
p_i^\top
\mathbb E_{\pi_{\mathbf\theta}}
\left[
\sum_{t=0}^{T}\gamma^t
\mathbf R_i(s_t,\mathbf a_t,s_{t+1})
\right] \\
&=
J_{\mathrm{team}}(\mathbf\theta)
+
p_i^\top\mathbf J_i(\mathbf\theta).
\end{aligned}
\]
For a fixed \(\mathbf p\), this scalarization induces a preference-conditioned stochastic game \(\mathcal G(\mathbf p)\), whose agents generally have different payoffs. We use Nash equilibrium~\citep{nash, fink} as the solution concept: at an equilibrium \(\mathbf\theta^{*}(\mathbf p)\), no agent can improve its own payoff \(U_i\) by unilaterally changing its policy within the considered policy class while the other agents' policies remain fixed.

However, an equilibrium under a fixed preference profile does not necessarily maximize the final team objective. The key problem is therefore to coordinate
the preference profile so that the induced equilibrium achieves high team performance. Motivated by the social-welfare view of evaluating collective outcomes~\citep{game1} and the optimal-equilibrium perspective in team Markov games~\citep{game2}, we formulate the goal as finding a preference profile \(\mathbf p\) and an induced equilibrium \(\mathbf\theta\) that maximize the team objective \(J_{\mathrm{team}}\).

\section{Theoretical Analysis}
\label{sec:theoretical_analysis}

We now provide a theoretical explanation for coordinated preference learning. First, we characterize how a policy update affects the team objective and show that preference diversity can contribute a positive first-order improvement. Second, under suitable regularity and second-order conditions, we show that strict local Nash equilibria under nearby preference profiles are connected through a locally smooth path, which implies that preference optimization can proceed gradually from an initial profile instead of treating each preference-conditioned game as an isolated problem.

\subsection{First-order team improvement and preference diversity}
\label{sec:first_order_team_improvement}
In this subsection, we ask how the choice of preference profile \((p_1, \cdots, p_N)\) affects the resulting team improvement. Starting from the current policy parameters \(\mathbf\theta=(\theta_1,\ldots,\theta_N)\), each agent performs one preference-conditioned gradient update
\[
\theta_{i,\mathrm{new}}
=
\theta_i
+
\eta
\nabla_{\theta_i}
U_i(\mathbf\theta;p_i),
\qquad i=1,\ldots,N ,
\]
We can define the following team-improvement matrix \(B\in\mathbb R^{N\times K}\) by
\[
\label{def:b_ik}
B_{i,k}
:=
\left(
\nabla_{\theta_i}J_{\mathrm{team}}(\mathbf\theta)
\right)^\top
\nabla_{\theta_i}J_{i,k}(\mathbf\theta),
\qquad
i=1,\ldots,N,\; k=1,\ldots,K .
\]
Here, \(B_{i,k}\) measures the first-order contribution of agent \(i\)'s \(k\)-th individual objective to the team objective. Let \(b_i=(B_{i,1},\ldots,B_{i,K})^\top\in\mathbb R^K\) denote the \(i\)-th row of the team-improvement matrix \(B\), represented as a column vector. Let \(\bar p=N^{-1}\sum_{i=1}^{N}p_i\) and \(\bar b=N^{-1}\sum_{i=1}^{N}b_i\) be the averages across agents, and let \(\tilde p_i=p_i-\bar p\) and \(\tilde b_i=b_i-\bar b\) be their centered counterparts.
Inspired by gradient-alignment principles in multi-objective optimization~\citep{desideri2012multiple,liu2021conflict}, we assume the following sufficient condition, which requires agents' deviations from the mean preference to align with their relative contributions to team improvement.

\begin{assumption}[Preference-improvement alignment]
\label{ass:preference_improvement_alignment}
There exists \(\kappa>0\) such that, for every agent \(i\) with \(\tilde p_i\neq 0\), the projection of \(\tilde b_i\) onto \(\tilde p_i\) is lower bounded by \(\kappa\), i.e.,
\[
\frac{\tilde p_i^\top \tilde b_i}{\|\tilde p_i\|_2^2}
\ge
\kappa,
\qquad i=1,\ldots,N .
\]
\end{assumption}

By first-order Taylor expansion, we can derive the following decomposition
\begin{theorem}[Team Improvement Decomposition]
\label{theo:team_improvement_decomposition}
Under Assumption~\ref{ass:preference_improvement_alignment}, the team improvement as $\eta\to0^+$ satisfies
\[
J_{\mathrm{team}}(\mathbf\theta_{\mathrm{new}})
-
J_{\mathrm{team}}(\mathbf\theta)
\ge
\eta
\sum_{i=1}^{N}
\left\|
\nabla_{\theta_i}J_{\mathrm{team}}(\mathbf\theta)
\right\|_2^2
+
\eta N
\left(
\bar p^\top \bar b
+
\kappa \mathcal D_p
\right)
+o(\eta)
\]
where
\[
\mathcal D_p
:=
\frac{1}{2N^2}
\sum_{i=1}^{N}
\sum_{j=1}^{N}
\|p_i-p_j\|_2^2 .
\]
is the pairwise preference distance across agents.
\end{theorem}
This theorem decomposes team improvement into three first-order terms and a higher-order remainder. The first term captures direct team-gradient improvement. It is non-negative but can be noisy under sparse team rewards. The second term, \(\eta N\bar p^\top \bar b\), captures average preference alignment. The third term, \(\eta N\kappa\mathcal D_p\), captures the additional improvement from aligned preference diversity.

\subsection{Equilibrium Tracking}
For each preference profile \(\mathbf p\), there is a corresponding
preference-conditioned game \(\mathcal G(\mathbf p)\). A natural question is how the solutions of these games are related as \(\mathbf p\) changes. In this subsection, we first show that, under regularity and second-order conditions, a strict local Nash equilibrium changes continuously with respect to \(\mathbf p\). Then, when the preference profile varies slowly, gradient updates can track the corresponding moving equilibrium path.

By applying the implicit function theorem \citep{IFT} and the second-order sufficient conditions for a strict local Nash equilibrium~\citep{DNE}, we obtain the following local continuity result. The full proof is provided in Appendix~\ref{pf:continual_equilibrium}.
\begin{lemma}[Continuity of preference-conditioned local Nash equilibria]
\label{lem:local_continuity_stationary_solution}
Suppose that \(\mathbf\theta^{*}\in\operatorname{int}\Theta\) is a strict local Nash equilibrium of
\(\mathcal G(\mathbf p)\) satisfying
\[
\nabla_{\theta_i}U_i(\mathbf\theta^{*};p_i)=0,
\qquad
\nabla^2_{\theta_i\theta_i}U_i(\mathbf\theta^{*};p_i)\prec0,
\qquad i=1,\ldots,N .
\]
Assume that the joint gradient
\((\nabla_{\theta_1}U_1,\ldots,\nabla_{\theta_N}U_N)\) is continuously differentiable near \((\mathbf\theta^{*},\mathbf p)\), and its Jacobian with respect to \(\mathbf\theta\) is nonsingular at this point. Then there exist neighborhoods \(\mathcal U\) of \(\mathbf p\) and \(\mathcal V\subset\Theta\) of \(\mathbf\theta^{*}\), and a unique continuously differentiable mapping \(\mathbf\theta(\cdot):\mathcal U\to\mathcal V\) with \(\mathbf\theta(\mathbf p)=\mathbf\theta^{*}\), such that, for every \(\tilde{\mathbf p}\in\mathcal U\), \(\mathbf\theta(\tilde{\mathbf p})\) is the unique stationary point in \(\mathcal V\) and a strict local Nash equilibrium of \(\mathcal G(\tilde{\mathbf p})\).
\end{lemma}

The previous lemma guarantees a locally smooth branch of strict local Nash equilibria. For learning dynamic analysis, we further impose the following two regularity assumptions.
\begin{assumption}[Attraction property near equilibrium]
\label{ass:attraction}
Let \(\mathbf\theta(\mathbf p)\) denote a strict local Nash equilibrium under the preference profile \(\mathbf p\). Assume that there exist \(\rho\in(0,1)\) and a neighborhood of \(\mathbf\theta(\mathbf p)\) such that
\[
\|\mathbf\theta_{\mathrm{new}}-\mathbf\theta(\mathbf p)\|
\le
\rho
\|\mathbf\theta_{\mathrm{old}}-\mathbf\theta(\mathbf p)\|,
\]
for all \(\mathbf\theta_{\mathrm{old}}\) in this neighborhood.
\end{assumption}

\begin{assumption}[Local Lipschitz continuity]
\label{ass:lipschitz}
Let \(\mathbf\theta(\cdot):\mathcal U\to\Theta\) be the strict local Nash equilibrium mapping induced by Lemma~\ref{lem:local_continuity_stationary_solution}. Assume that this mapping is Lipschitz continuous in \(\mathcal U\): there
exists \(C>0\) such that for any two preference profiles
\(\mathbf p,\mathbf p'\in\mathcal U\),
\[
\|\mathbf\theta(\mathbf p')-\mathbf\theta(\mathbf p)\|
\le
C\|\mathbf p'-\mathbf p\|.
\]
\end{assumption}

The following theorem shows that, when the preference profile changes slowly, the policy update remains close to the equilibrium associated with the new preference profile. See proof in Appendix~\ref{pf_tracking}.
\begin{theorem}[Equilibrium tracking]
\label{thm:tracking}
Consider the iterative policy update
\[
\theta_i^{t+1}
=
\theta_i^t
+
\eta
\nabla_{\theta_i}U_i(\theta^t;\mathbf p^t),
\qquad i=1,\ldots,N .
\]
Under Assumption~\ref{ass:attraction} and Assumption~\ref{ass:lipschitz}, the tracking error
$
e_t
:=
\|\theta^t-\theta(\mathbf p^t)\|
$
satisfies
\[
e_{t+1}
\le
\rho e_t
+
C\|\mathbf p^{t+1}-\mathbf p^t\|.
\]
In particular, if
\(\|\mathbf p^{t+1}-\mathbf p^t\|\le\delta\), then
\[
\limsup_{t\to\infty} e_t
\le
\frac{C}{1-\rho}\delta .
\]
\end{theorem}

\section{Multi-agent Policy Optimization with Preference Coordinated Learning}
\label{sec:methods}

\begin{figure}[t]
    \centering
    \includegraphics[width=1.0\linewidth]{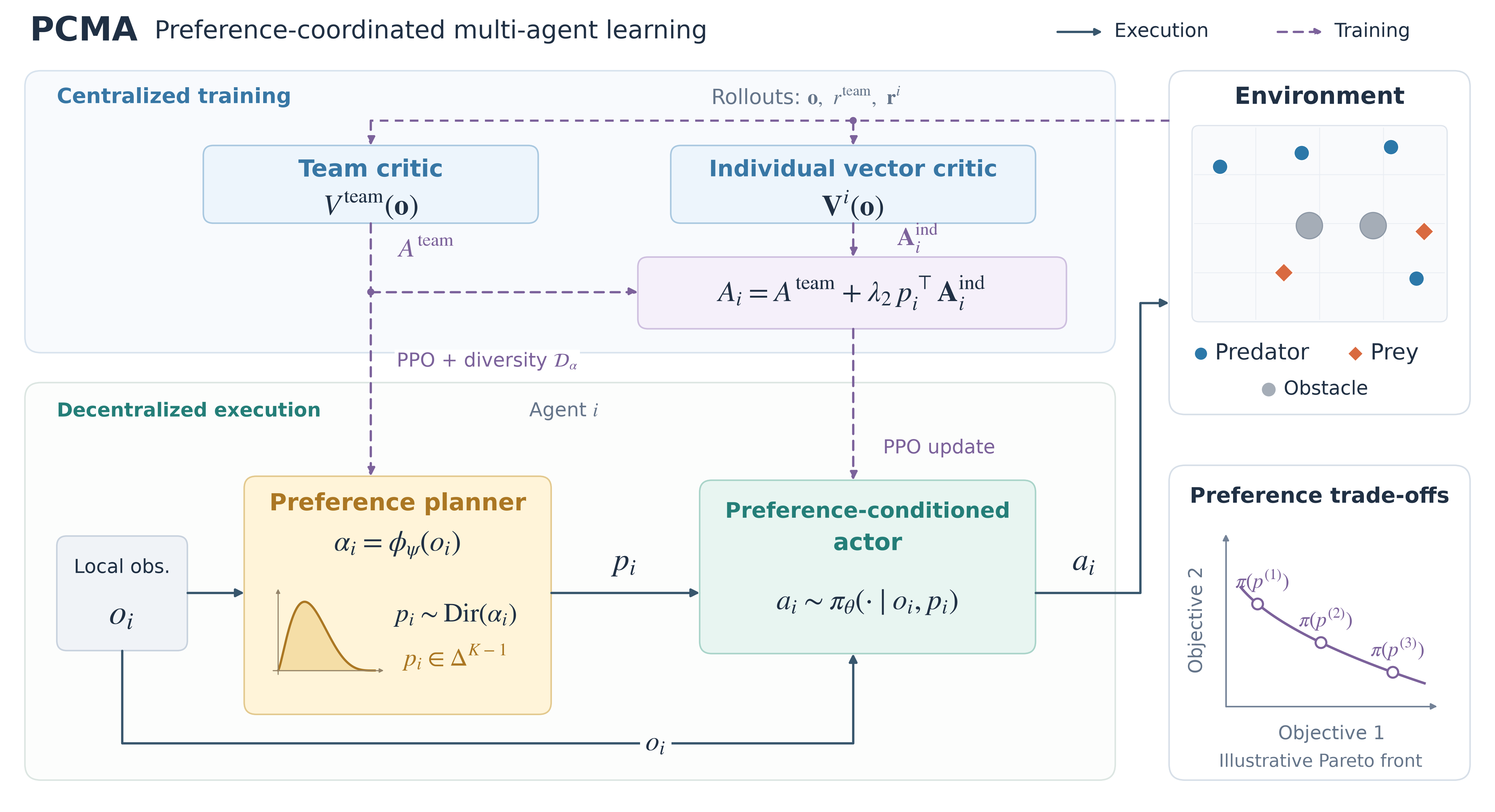}
    \caption{Overview of PCMA. Each agent uses a stochastic preference planner and a preference-conditioned actor. The planner samples preferences from a Dirichlet distribution, and the actor selects actions conditioned on the sampled preference. During training, the team critic provides coordination feedback, while individual vector critics provide dense preference-aligned learning signals.}
    \label{fig:algo-flowchart}
\end{figure}

Motivated by the preceding analysis, we now implement the preference-coordination idea in a practical PPO-based algorithm, called Preference Coordinated Multi-agent Policy Optimization (PCMA). The key design is to treat preference \(p_i\) as a latent coordination variable: it determines the utility optimized by agent \(i\), while the planner adjusts the preference profile to encourage team-beneficial updates. PCMA follows the CTDE \citep{MADDPG} paradigm, using centralized critics for training and decentralized policies for execution.

\subsection{Preference-conditioned Policy Optimization}
\label{sec:preference_conditioned_policy_optimization}

Following preference-conditioned policy learning in MORL~\citep{CN,envelope,PD-MORL}, each agent conditions its actor on its local observation \(o_i\) and preference \(p_i\), selecting \(a_i\sim\pi_\theta(\cdot\mid o_i,p_i)\).

The team critic \(V^{\mathrm{team}}(\mathbf o)\), with concatenated observations \(\mathbf o=(o_1,\ldots,o_N)\), estimates the expected discounted team return. Each agent \(i\) has its own vector critic \(\mathbf V^i(\mathbf o)\in\mathbb R^K\), which estimates the returns for its objectives. We use these value estimates and rewards to compute the scalar team advantage \(A^{\mathrm{team}}\) and individual advantage vector \(\mathbf A_i^{\mathrm{ind}}\) via generalized advantage estimation (GAE)~\citep{GAE}, applied componentwise to individual objectives. Agent \(i\) combines the team and preference-weighted individual advantages for its PPO update:
\begin{equation}
\label{eq:actor_loss}
\mathcal L_{\mathrm{actor}}(\theta)
=
\mathcal L_{\mathrm{PPO}}
\left(
\pi_\theta(\cdot\mid o_i,p_i),
A_i
\right),
\qquad
A_i
=
A^{\mathrm{team}}
+
\lambda_2 p_i^\top \mathbf A_i^{\mathrm{ind}} .
\end{equation}
Here, \(\lambda_2\) weights individual guidance relative to the team signal. We train the critics on discounted return targets \(R^{\mathrm{team}}\) and \(\mathbf R^i\) by minimizing
\begin{equation}
\label{eq:critic_loss}
\mathcal L_{\mathrm{critic}}
=
\mathbb E_{\tau}
\left[
\left(
V^{\mathrm{team}}(\mathbf o)-R^{\mathrm{team}}
\right)^2
+
\sum_{i=1}^{N}
\left\|
\mathbf V^{i}(\mathbf o)-\mathbf R^{i}
\right\|_2^2
\right],
\end{equation}
where the expectation is over collected rollouts \(\tau\). We next learn the preferences themselves to improve team coordination.

\subsection{Coordinated Preference Planning}
\label{sec:coordinated_preference_planning}

We then equip each agent with a stochastic preference planner that adapts its own preference based on its local observation. Since preferences lie on the simplex, a natural choice is the Dirichlet distribution. Specifically, the planner outputs \(\alpha_i=\phi_\psi(o_i)\), and samples \(p_i\sim\mathrm{Dir}(\alpha_i)\), where \(p_i\in\Delta^{K-1}\). Motivated by Theorem~\ref{theo:team_improvement_decomposition}, we encourage the planner to maintain diverse preference distributions across agents. We regularize the expected pairwise diversity of the sampled preferences \(\mathcal D_\alpha:=\mathbb E[\mathcal D_p]\). This regularizer promotes coordinated specialization: different agents can focus on different objective trade-offs, rather than collapsing to the same preference direction. The planner is then trained with the team advantage:
\begin{equation}
\label{eq:planner_loss}
\mathcal L_{\mathrm{plan}}(\psi)
=
\mathcal L_{\mathrm{PPO}}
\left(
\phi_\psi(\cdot\mid o_i),
A^{\mathrm{team}}
\right)
-
\lambda_1 \mathcal D_{\alpha}.
\end{equation}
The PPO term encourages preference distributions that yield higher team advantage. Therefore, our algorithm~\ref{alg:PCMA} not only let agent learn from local utilities, but are guide them towards team-level optimality.

\begin{algorithm}[t]
\caption{Preference Coordinated Multi-agent Policy Optimization(PCMA)}
\label{alg:PCMA}
\begin{algorithmic}[1]
\State Initialize preference planner \(\phi_\psi\), actor \(\pi_\theta\), team critic
\(V^{\mathrm{team}}\), and vector critics \(\{\mathbf V^i\}_{i=1}^N\)
\For{each iteration}
    \For{each environment step}
        \For{each agent \(i=1,\ldots,N\)}
            \State sample local preference \(p_i\sim\mathrm{Dir}(\phi_{\psi}(o_i))\) and Sample action \(a_i\sim\pi_\theta(\cdot\mid o_i,p_i)\)
        \EndFor
        \State Execute joint action \(\mathbf a=(a_1,\ldots,a_N)\) and store
        \((\mathbf o,\mathbf p,\mathbf a,r^{\mathrm{team}},\{\mathbf r^i\}_{i=1}^N,\mathbf o')\)
    \EndFor
    \State Estimate the team advantage \(A^{\mathrm{team}}\) and individual vector advantages \(\{\mathbf A_i^{\mathrm{ind}}\}_{i=1}^N\)
    \State Update \(V^{\mathrm{team}}\) and \(\{\mathbf V^i\}_{i=1}^N\) by minimizing $\mathcal{L}_{\text{critic}}$ defined in Eq.~\eqref{eq:critic_loss}
    \State Compute $A_i$ and Update \(\pi_\theta\) by minimizing the $\mathcal{L}_{\text{actor}}$ defined in Eq.~\eqref{eq:actor_loss}
    \State Compute preference diversity \(\mathcal D_\alpha\), and update \(\phi_\psi\) by minimizing the $\mathcal{L}_{\text{plan}}$ defined in Eq.~\eqref{eq:planner_loss}
\EndFor
\end{algorithmic}
\end{algorithm}

\section{Experiments}
\label{sec:experiments}

In this section, we present empirical results to validate the effectiveness of our method. We evaluate on a diverse set of cooperative multi-agent environments, including particle-world coordination~\citep{MADDPG}, drone control~\citep{crazy_rl}, walker locomotion~\citep{petting_zoo}, and StarCraft combat~\citep{SMAC}. To better study coordinated preference learning, we modify these environments by separating the sparse team reward from vector-valued individual rewards. This design allows us to test whether our method can organize vectorized local learning signals toward improved team performance. The numbers of agents and objectives for each environment, together with detailed environment settings, are provided in Appendix~\ref{subsec:env_details} (Table~\ref{tab:env_settings}). We report team success rate and mean reward across agents and objectives in Table~\ref{tab:main_results}. Additional PCMA-SAC results on OpenCDA-MARL/CARLA are reported in Appendix~\ref{app:opencda_results}.

\subsection{How Does Preference Coordination Improve Cooperation?}
\label{sec:pref_coord_analysis}

\paragraph{Preference specialization in particle-world tasks.}
We first analyze preference specialization in MOMPE tasks. In Cooperative Spread, we use objective-specific distance rewards, where each objective encourages agents to approach a different landmark. In Predator-Prey, the objectives correspond to approaching different prey while maintaining safety.

Cooperative Spread has two agents and two objectives: the negative distances to landmarks 1 and 2 (OBJ0 and OBJ1, respectively). Figure~\ref{fig:analysis_spread_pref} shows that, as training progresses, the two agents learn preferences for different landmarks. This specialization encourages them to cover distinct landmarks rather than crowd around the same one. Safe Predator-Prey has four predators and three objectives: progress toward prey 1 (OBJ0), progress toward prey 2 (OBJ1), and safety (OBJ2). Figure~\ref{fig:analysis_predator_pref} similarly shows that the four predators learn different preferences. A0 favors prey 1, while A3 favors prey 2. A1 and A2 balance pursuit and safety more evenly. Figures~\ref{fig:analysis_spread_pf} and~\ref{fig:analysis_predator_pf} show that agents' returns are distributed along the corresponding return fronts.


\begin{figure}[htbp]
    \centering
    \captionsetup[subfigure]{skip=2pt,font=scriptsize}
    \begin{subfigure}[htbp]{0.24\linewidth}
        \centering
        \includegraphics[width=\linewidth]{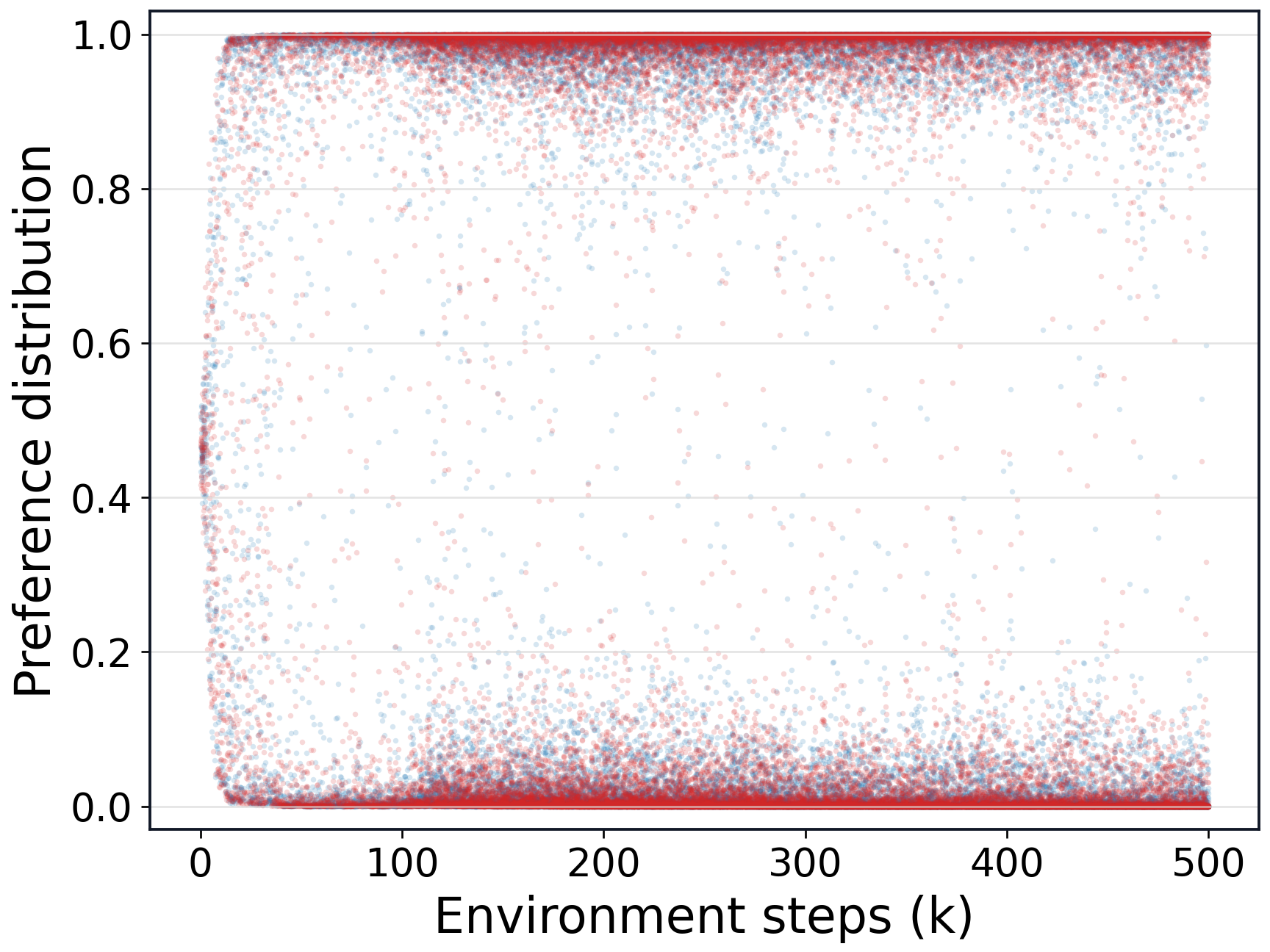}
        \caption{Preferences (Spread).}
        \label{fig:analysis_spread_pref}
    \end{subfigure}
    \hfill
    \begin{subfigure}[htbp]{0.24\linewidth}
        \centering
        \includegraphics[width=\linewidth]{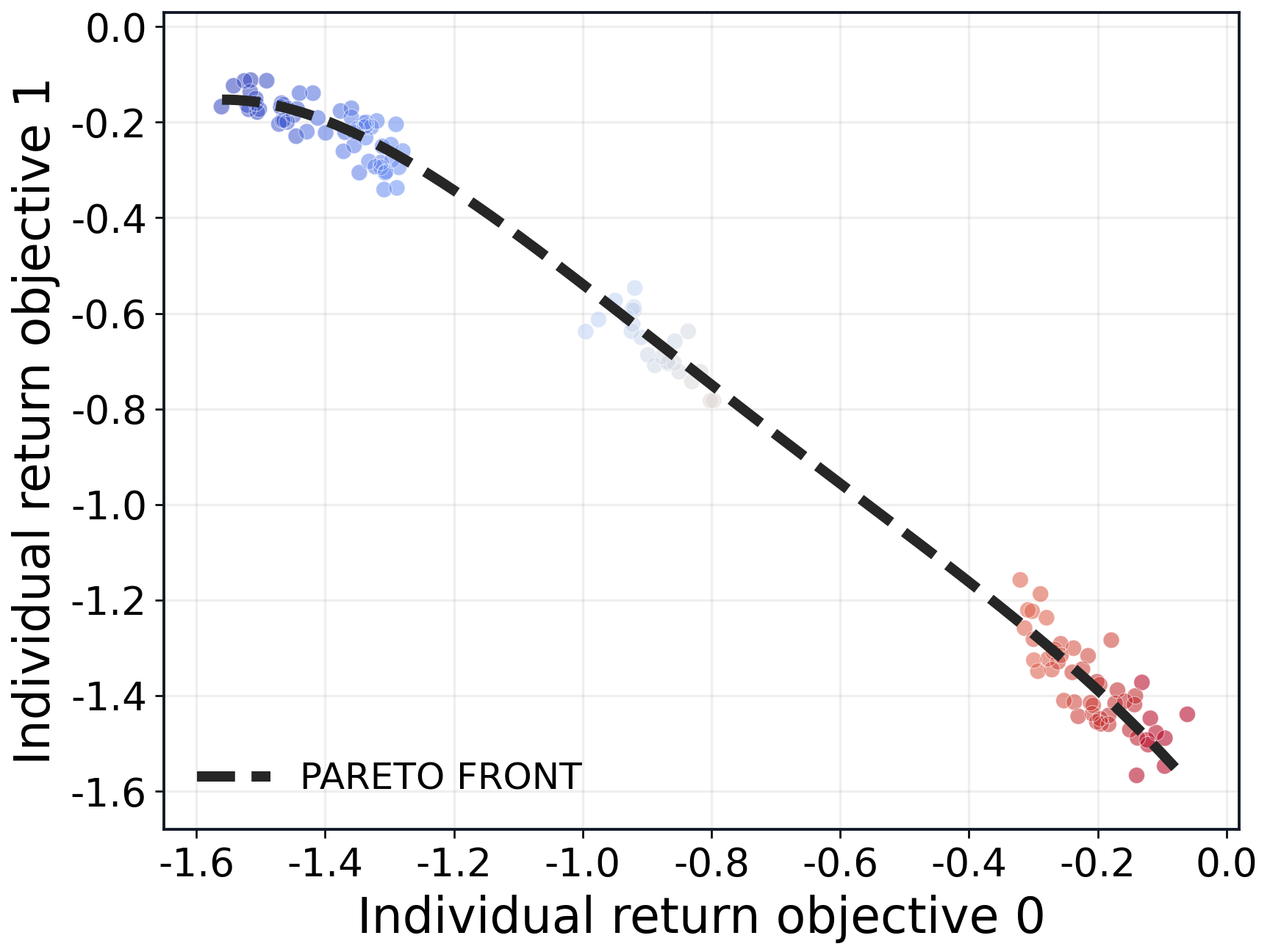}
        \caption{Return front (Spread).}
        \label{fig:analysis_spread_pf}
    \end{subfigure}
    \hfill
    \begin{subfigure}[htbp]{0.24\linewidth}
        \centering
        \includegraphics[width=\linewidth]{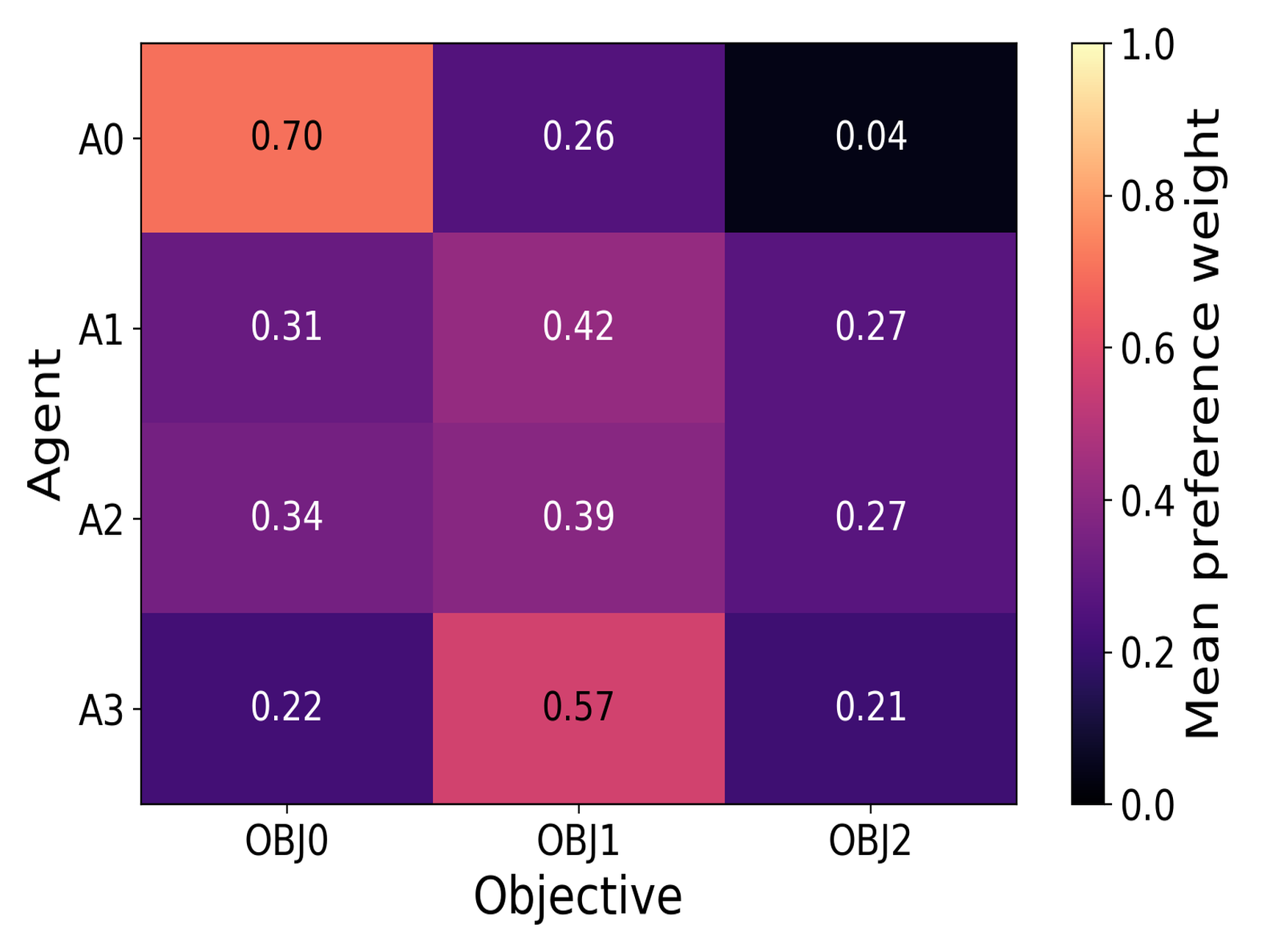}
        \caption{Preferences (Predator).}
        \label{fig:analysis_predator_pref}
    \end{subfigure}
    \hfill
    \begin{subfigure}[htbp]{0.24\linewidth}
        \centering
        \includegraphics[width=\linewidth]{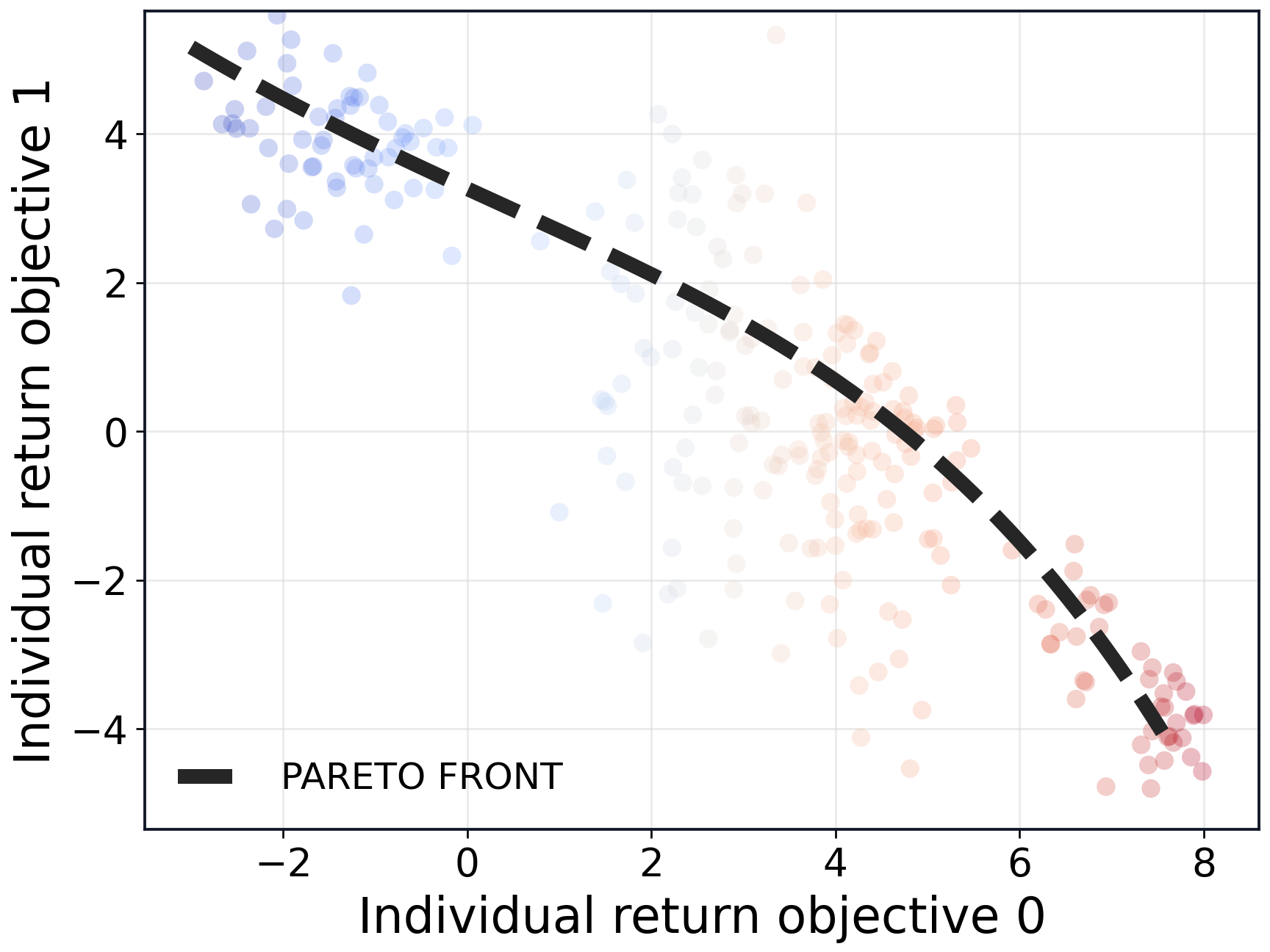}
        \caption{Return front (Predator).}
        \label{fig:analysis_predator_pf}
    \end{subfigure}

    \caption{
    Learned preferences and return fronts for PCMA in particle-world tasks.
    (\subref{fig:analysis_spread_pref}) Preference weights during Cooperative Spread training.
    (\subref{fig:analysis_spread_pf}) 2D vector returns along the front for the two landmark objectives.
    (\subref{fig:analysis_predator_pref}) Mean preference weights in Safe Predator-Prey.
    (\subref{fig:analysis_predator_pf}) 2D vector returns along the front for the two prey objectives.
    }
    \label{fig:mompe_pref_analysis}
\end{figure}
\paragraph{Role differentiation in SMAC.}
The same effect appears in SMAC. Figure~\ref{fig:smac_role_analysis} measures
each agent by damage dealt and damage taken, which reflect different combat roles such as aggressive attackers or front-line units.
With coordinated preferences, agents spread across the behavior space. This indicates that preference coordination induces role differentiation: different agents emphasize different combat trade-offs while still contributing to the same team objective. Thus, coordinated
preference diversity provides a simple mechanism for improving cooperation.

\begin{figure}[!htb]
    \centering
    \captionsetup[subfigure]{skip=2pt}

    \begin{subfigure}[t]{0.34\linewidth}
        \centering
        \includegraphics[width=\linewidth]{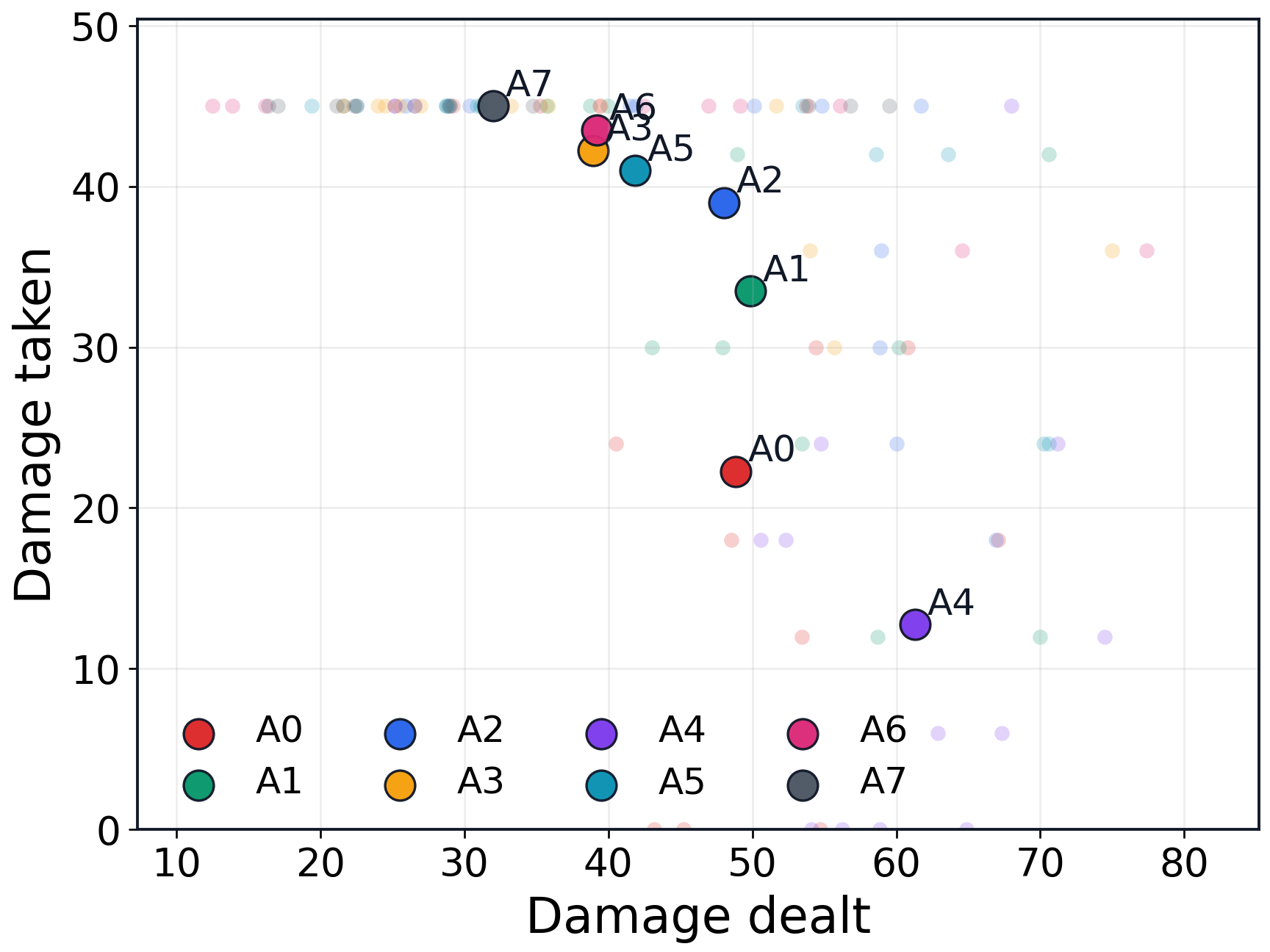}
        \caption{\texttt{8m}.}
        \label{fig:analysis_smac_8m}
    \end{subfigure}
    \hspace{0.04\linewidth}
    \begin{subfigure}[t]{0.34\linewidth}
        \centering
        \includegraphics[width=\linewidth]{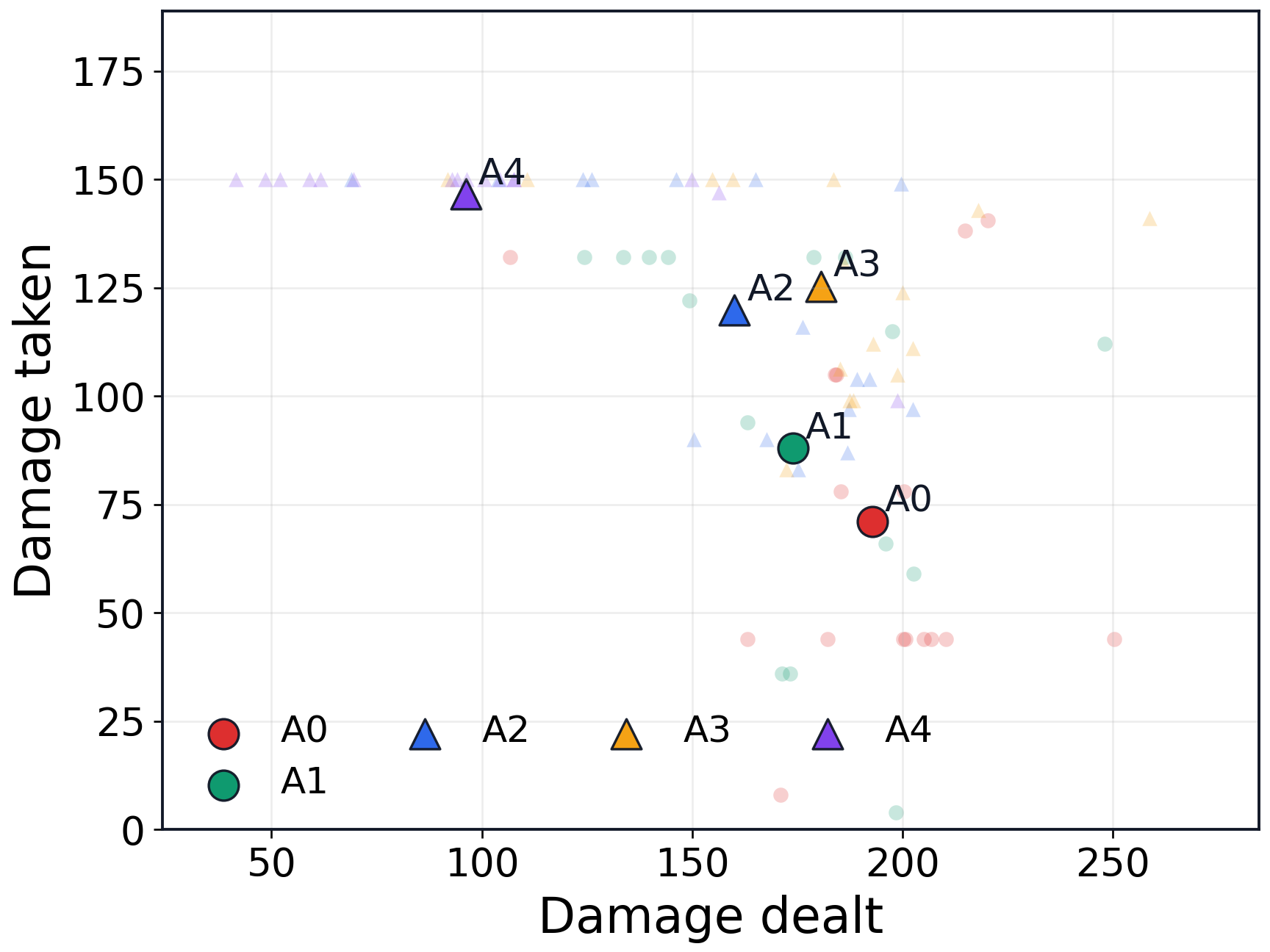}
        \caption{\texttt{2s3z}.}
        \label{fig:analysis_smac_2s3z}
    \end{subfigure}

    \caption{
    Combat behavior under PCMA in (\subref{fig:analysis_smac_8m}) \texttt{8m} and (\subref{fig:analysis_smac_2s3z}) \texttt{2s3z}. The horizontal and vertical axes show damage dealt and damage taken. Large labeled markers show each agent's episode-averaged values.}
    \label{fig:smac_role_analysis}
\end{figure}

\subsection{Quantitative Performance}
\label{sec:quantitative_results}
We compare PCMA with MADDPG, IPPO, and MAPPO, covering deterministic actor-critic learning, independent policy optimization, and policy optimization shared critic. All baselines use uniform scalarization weights for the vector-valued guiding reward. More details are shown in Appendix~\ref{subsec:training_details}. For most environments, we report success rate and average reward as our main metric in Table~\ref{tab:main_results}. For Escort and MOMAwalker, where success rate is not directly defined, we instead report task-specific metrics such as average reward and forward distance. Fig.~\ref{fig:learning_curves} shows the learning curves across the evaluated tasks. PCMA achieves the best or tied-best performance on most metrics, showing consistent improvements across diverse cooperative multi-agent tasks.

\begin{figure*}[htbp] 
    \centering
    \begin{subfigure}[b]{0.24\textwidth}
        \centering
        \includegraphics[width=\linewidth]{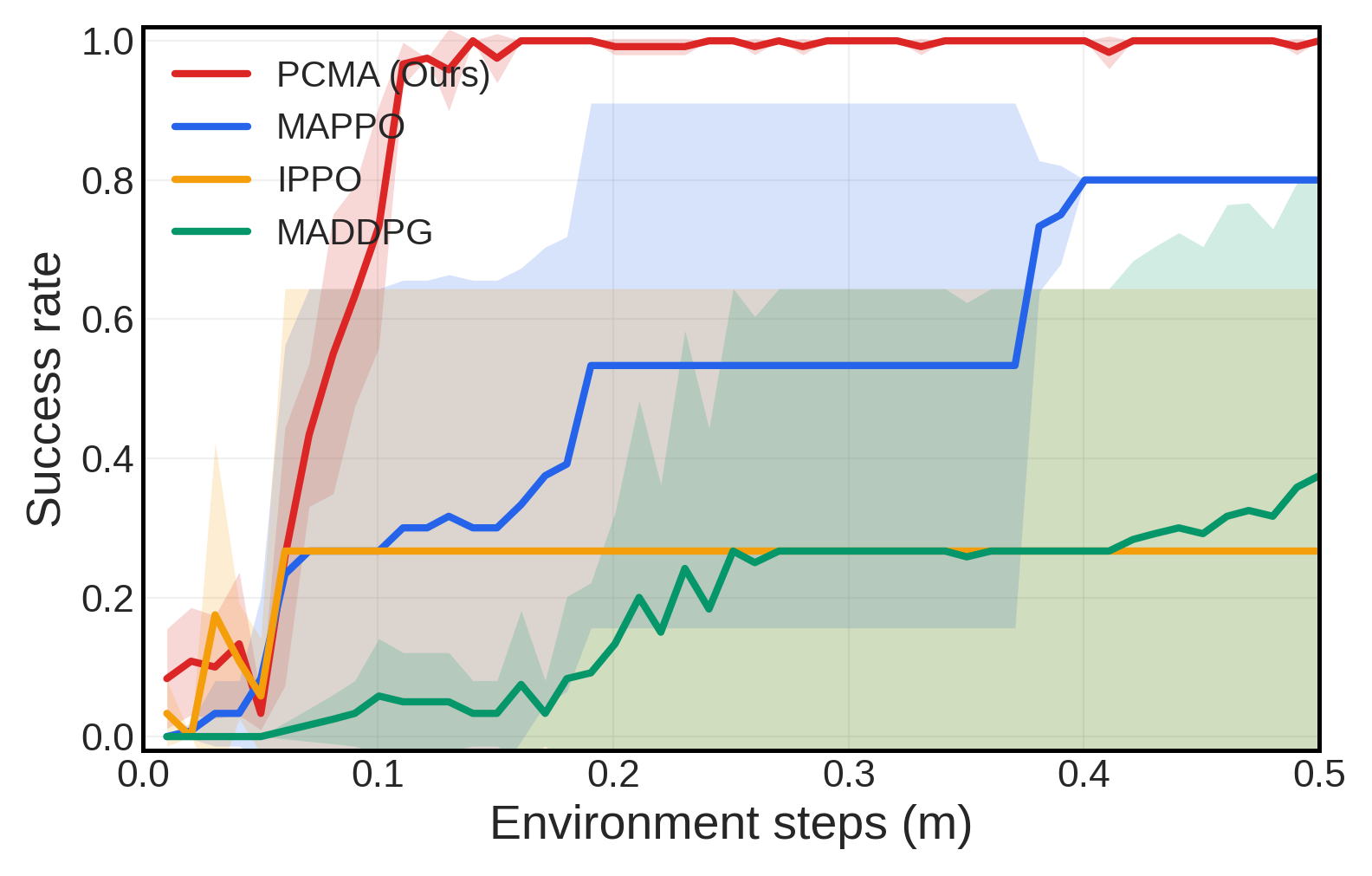}
        \caption{Cooperative Spread}
        \label{fig:results_a}
    \end{subfigure}
    \hfill
    \begin{subfigure}[b]{0.24\textwidth}
        \centering
        \includegraphics[width=\linewidth]{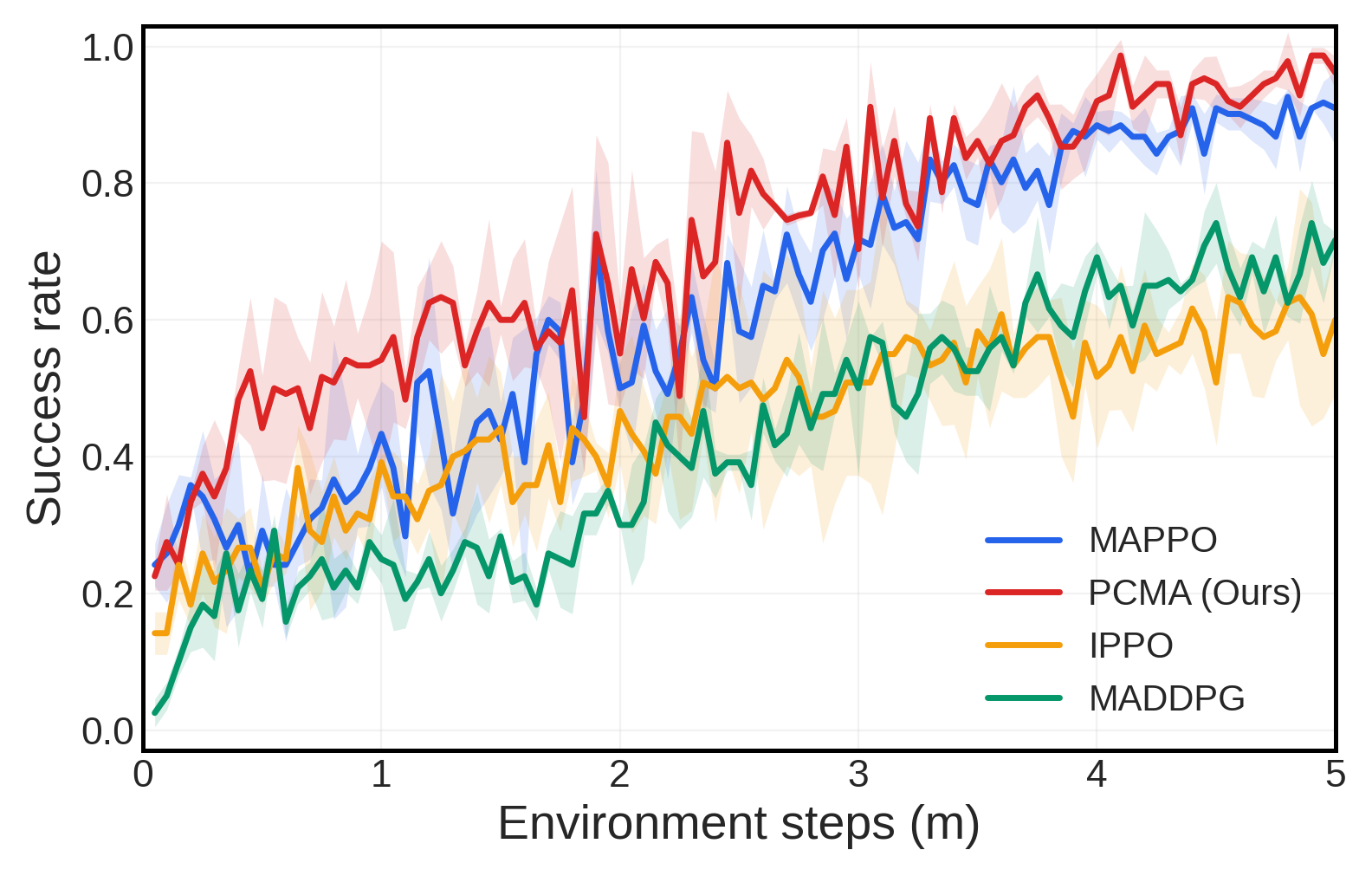}
        \caption{Safe Predator Prey}
        \label{fig:results_b}
    \end{subfigure}
    \hfill
    \begin{subfigure}[b]{0.24\textwidth}
        \centering
        \includegraphics[width=\linewidth]{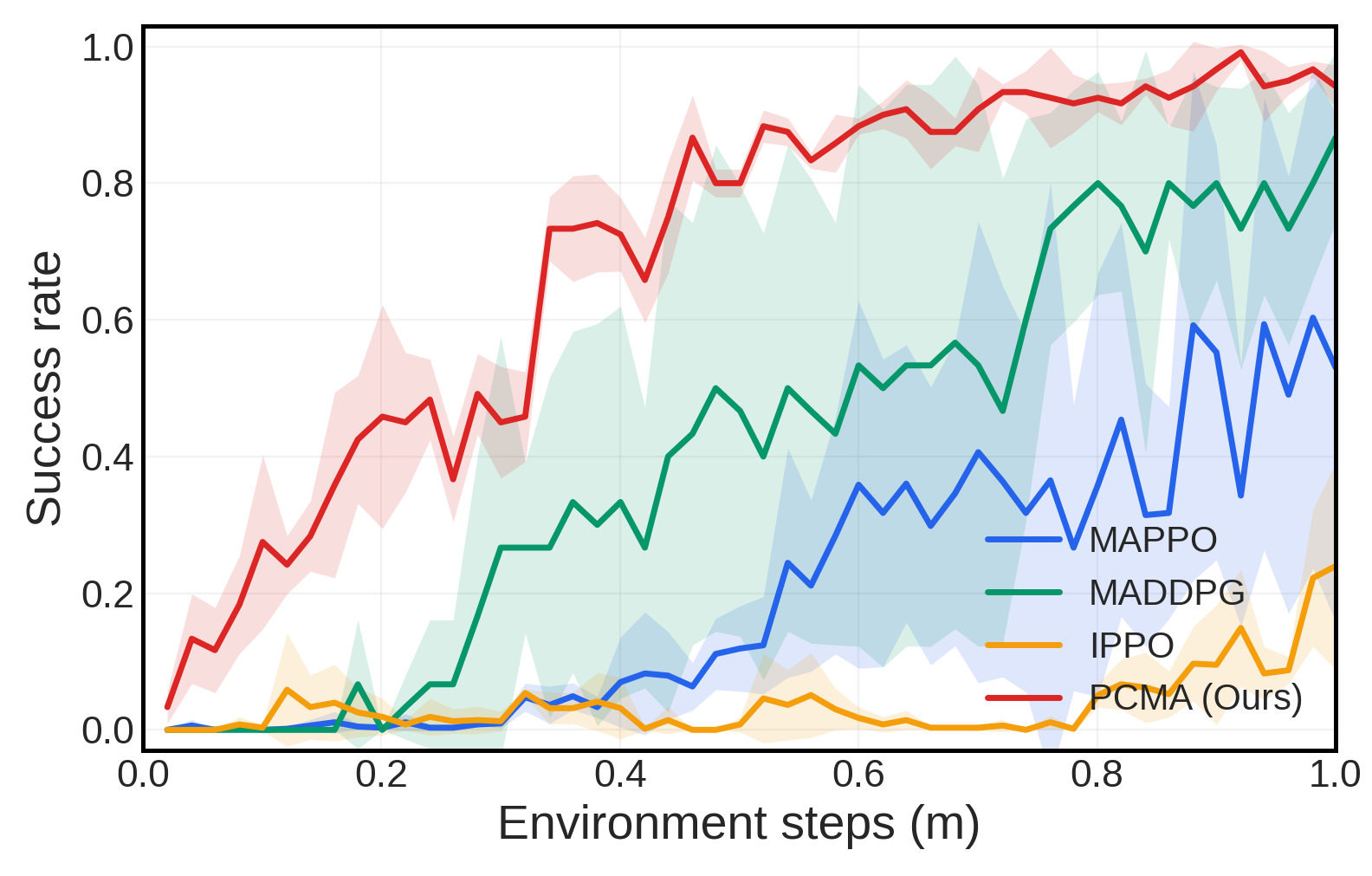}
        \caption{Catch}
        \label{fig:results_c}
    \end{subfigure}
    \hfill
    \begin{subfigure}[b]{0.24\textwidth}
        \centering
        \includegraphics[width=\linewidth]{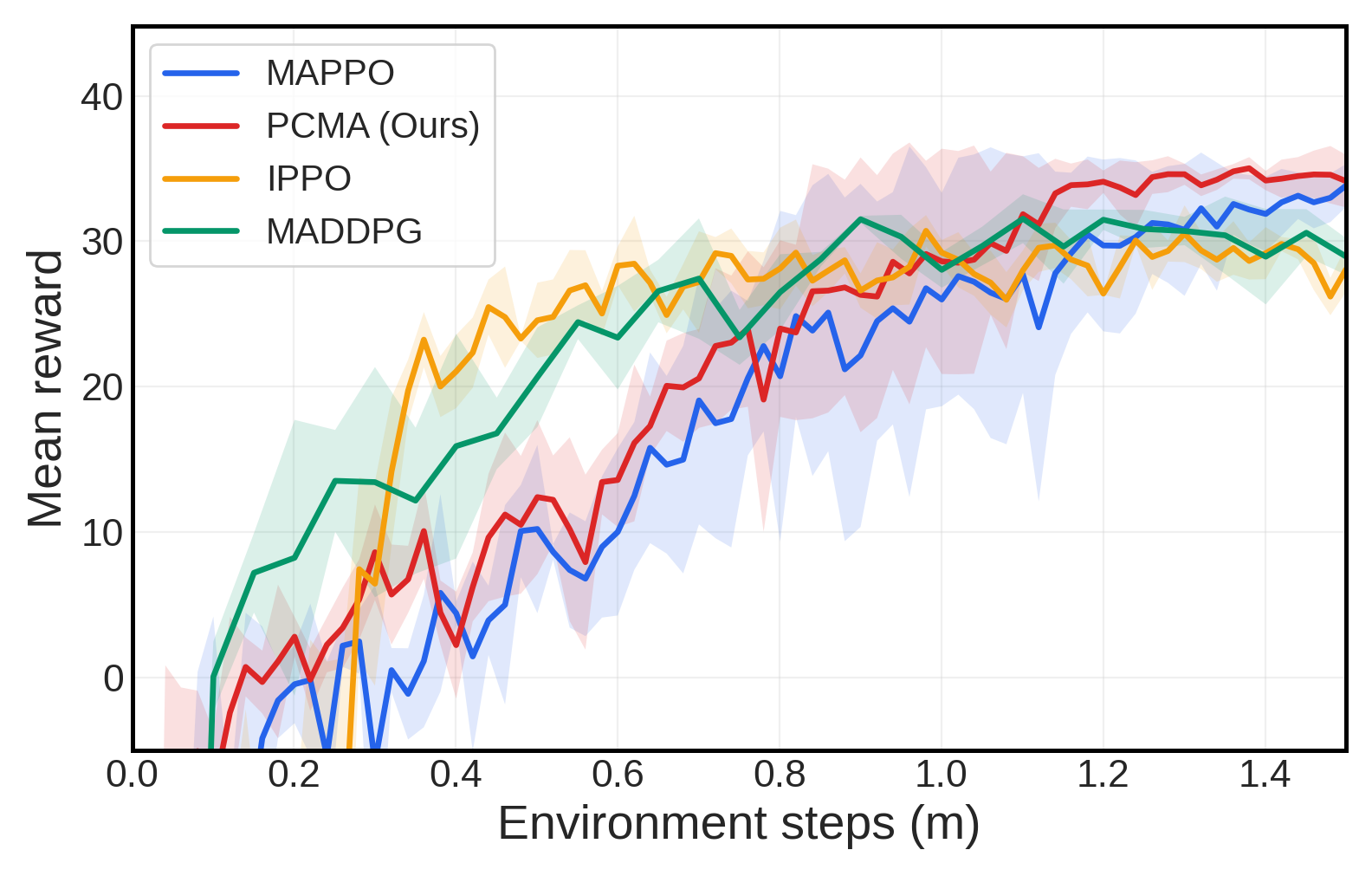}
        \caption{Escort}
        \label{fig:results_d}
    \end{subfigure}

    \vspace{0.8em} 

    \begin{subfigure}[b]{0.24\textwidth}
        \centering
        \includegraphics[width=\linewidth]{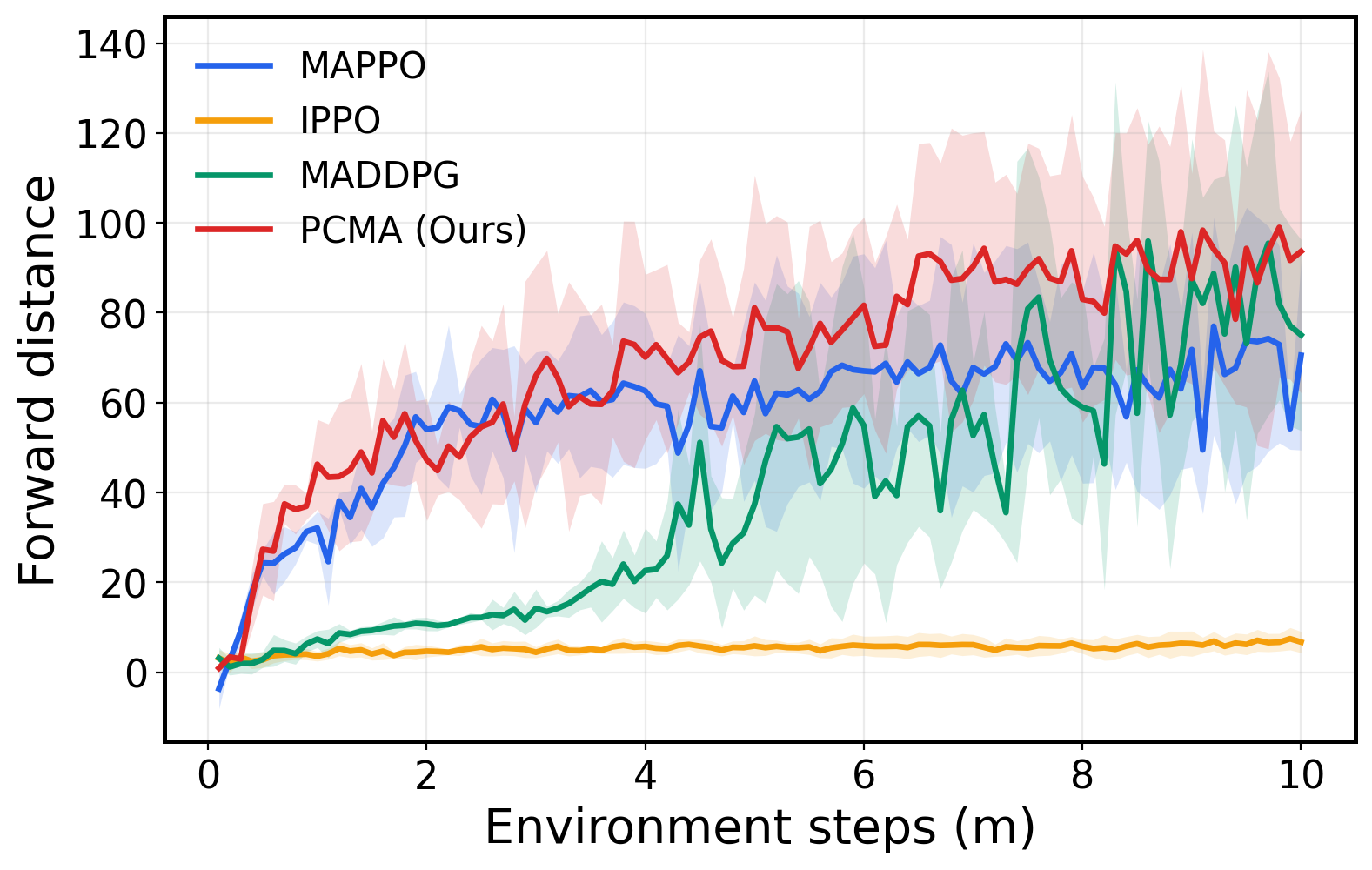}
        \caption{Walker 1}
        \label{fig:results_e}
    \end{subfigure}
    \hfill
    \begin{subfigure}[b]{0.24\textwidth}
        \centering
        \includegraphics[width=\linewidth]{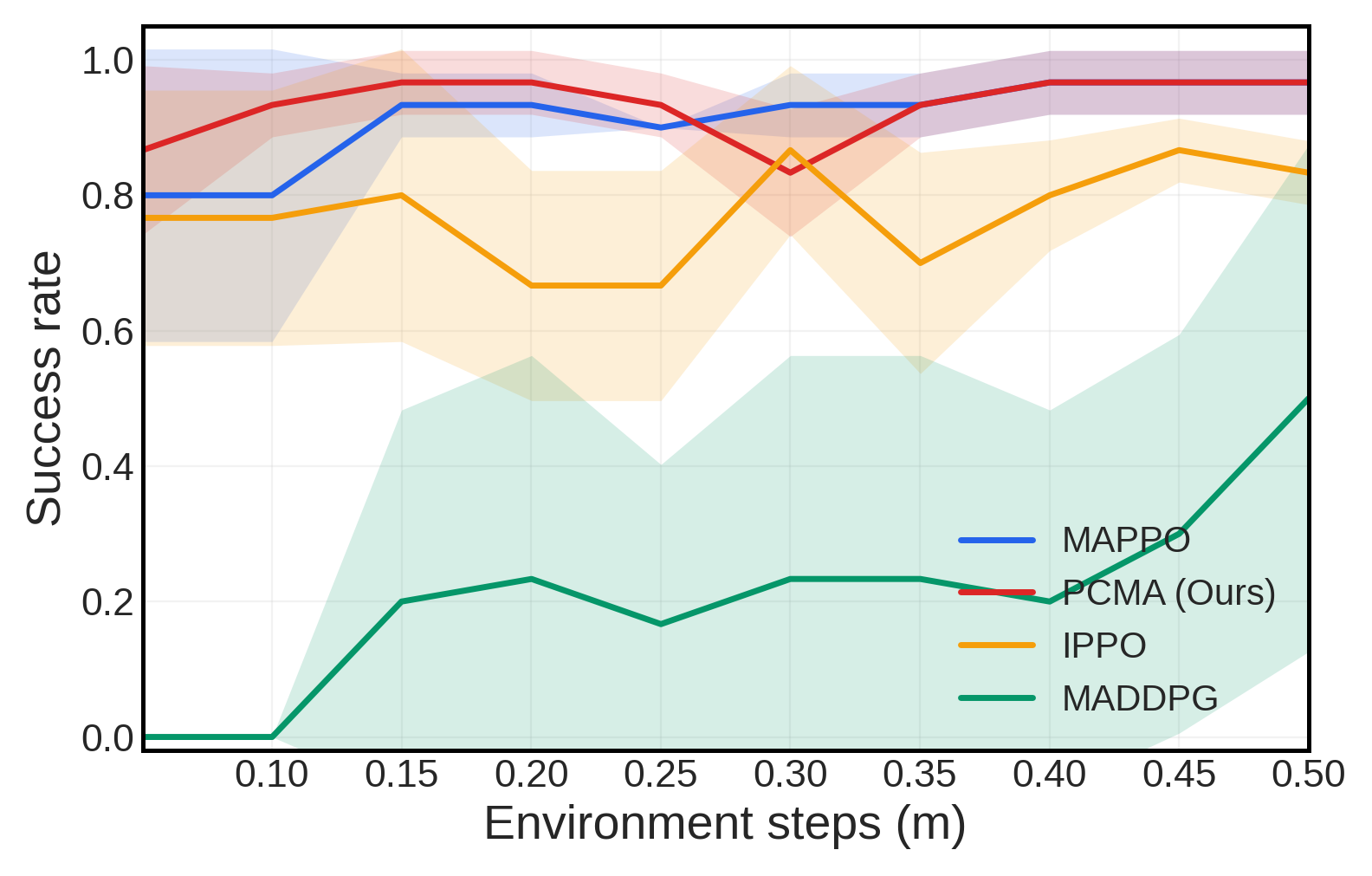}
        \caption{SMAC 3m}
        \label{fig:results_f}
    \end{subfigure}
    \hfill
    \begin{subfigure}[b]{0.24\textwidth}
        \centering
        \includegraphics[width=\linewidth]{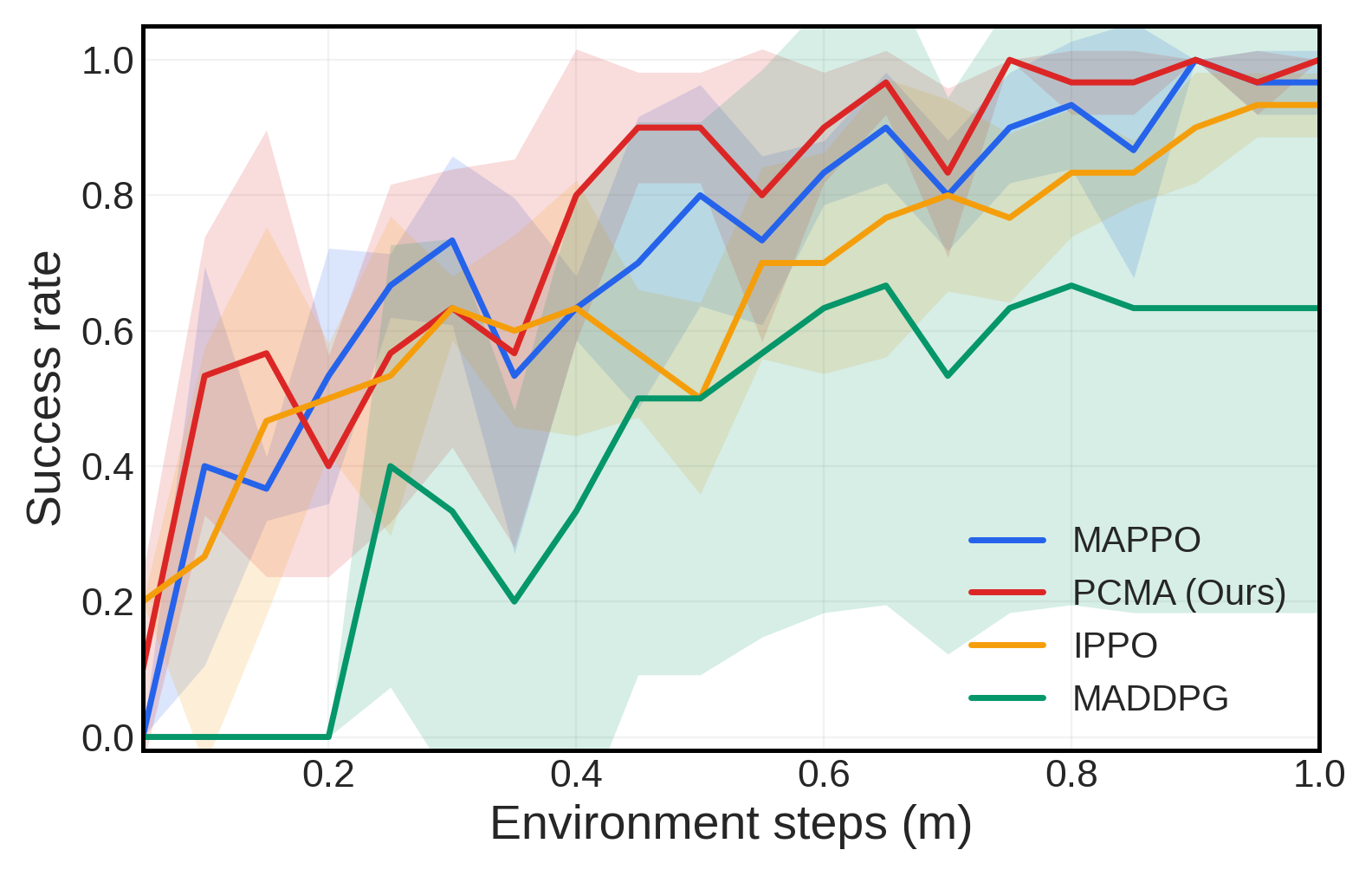}
        \caption{SMAC 2s3z}
        \label{fig:results_g}
    \end{subfigure}
    \hfill
    \begin{subfigure}[b]{0.24\textwidth}
        \centering
        \includegraphics[width=\linewidth]{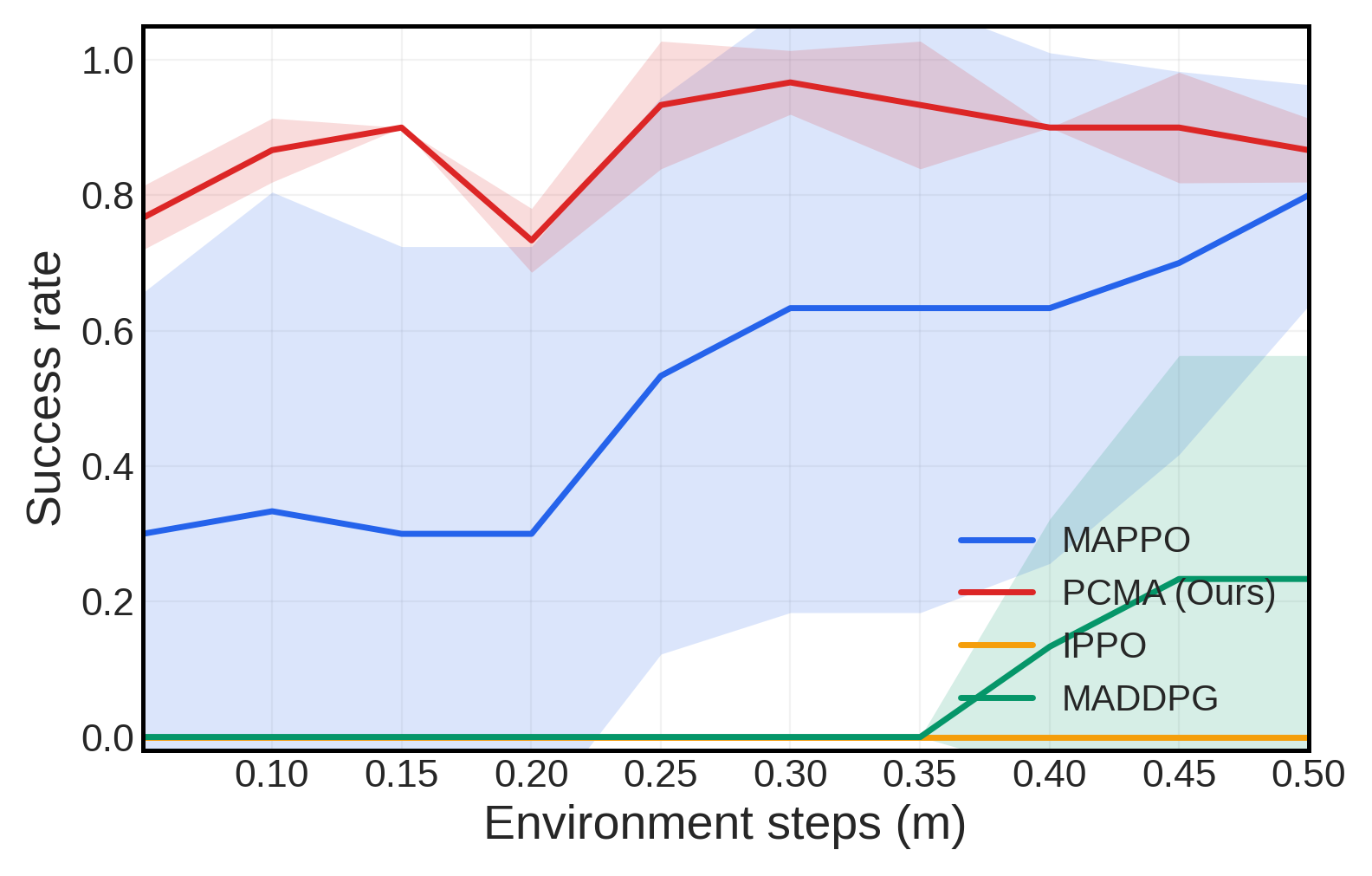}
        \caption{SMAC 8m}
        \label{fig:results_h}
    \end{subfigure}
    \caption{
    Learning curves of PCMA and baseline methods across the evaluated cooperative multi-agent tasks. Curves are averaged over 3 random seeds, and shaded regions indicate standard deviation.
    }
    \label{fig:learning_curves}
\end{figure*}
\begin{table}[!htb]
\centering
\caption{Results on multi-agent tasks. Values in parentheses denote standard deviation.}
\label{tab:main_results}
\setlength{\tabcolsep}{3.5pt} 
\small 
\begin{tabular}{llrrrr}
\toprule
Environment & Metric & MADDPG & IPPO & MAPPO & PCMA (Ours) \\
\midrule
\multirow{2}{*}{Cooperative Spread}
& Success rate
& $0.38_{(0.43)}$
& $0.27_{(0.38)}$
& $0.80_{(0.00)}$
& $\mathbf{1.00}_{(0.00)}$ \\
& Average reward
& $-0.63_{(1.09)}$
& $0.20_{(0.35)}$
& $0.69_{(0.00)}$
& $\mathbf{0.89}_{(0.00)}$ \\
\midrule
\multirow{2}{*}{Safe Predator Prey} 
& Success rate   
& $0.68_{(0.06)}$ 
& $0.60_{(0.01)}$ 
& $0.91_{(0.05)}$ 
& $\mathbf{0.96}_{(0.02)}$ \\
& Average reward 
& $2.30_{(0.10)}$ 
& $2.21_{(0.02)}$ 
& $2.39_{(0.04)}$ 
& $\mathbf{2.45}_{(0.05)}$ \\
\midrule
\multirow{2}{*}{Catch} 
& Success rate   
& $0.87_{(0.12)}$
& $0.24_{(0.15)}$ 
& $0.53_{(0.37)}$  
& $\mathbf{0.94}_{(0.03)}$  \\
& Average reward 
& $4.71_{(7.94)}$ 
& $8.61_{(0.64)}$ 
& $11.33_{(3.83)}$ 
& $\mathbf{14.21}_{(0.83)}$ \\
\midrule
Escort
& Average reward
& $29.19_{(0.80)}$
& $27.74_{(0.74)}$
& $33.72_{(0.63)}$
& $\mathbf{34.27}_{(0.99)}$ \\
\midrule
\multirow{2}{*}{MOMAwalker}
& Forward Distance
& $75.04_{(21.28)}$
& $6.69_{(2.36)}$
& $70.52_{(21.08)}$
& $\mathbf{93.64}_{(31.47)}$ \\
& Average reward
& $18.28_{(16.18)}$
& $-0.80_{(0.37)}$
& $10.11_{(15.35)}$
& $\mathbf{21.62}_{(20.05)}$ \\
\midrule
\multirow{2}{*}{SMAC-3m}
& Success rate
& $0.50_{(0.37)}$
& $0.83_{(0.05)}$
& $\mathbf{0.97}_{(0.05)}$
& $\mathbf{0.97}_{(0.05)}$ \\
& Average reward
& $0.75_{(0.31)}$
& $1.02_{(0.03)}$
& $0.29_{(0.02)}$
& $\mathbf{1.11}_{(0.03)}$ \\
\midrule
\multirow{2}{*}{SMAC-2s3z}
& Success rate
& $0.63_{(0.45)}$
& $0.93_{(0.05)}$
& $0.97_{(0.05)}$
& $\mathbf{1.00}_{(0.00)}$ \\
& Average reward
& $0.22_{(0.10)}$
& $0.29_{(0.01)}$
& $0.30_{(0.01)}$
& $\mathbf{0.31}_{(0.01)}$ \\
\midrule
\multirow{2}{*}{SMAC-8m}
& Success rate
& $0.23_{(0.33)}$
& $0.00_{(0.00)}$
& $0.80_{(0.16)}$
& $\mathbf{0.87}_{(0.05)}$ \\
& Average reward
& $0.08_{(0.12)}$
& $0.00_{(0.00)}$
& $0.27_{(0.04)}$
& $\mathbf{0.28}_{(0.01)}$ \\

\bottomrule
\end{tabular}
\end{table}

\subsection{Ablation Study}
\label{sec:ablation}
As shown in Figure ~\ref{fig:ablation}, we conduct ablation studies on Cooperative Spread to analyze three key factors in PCMA: the diversity regularization coefficient \(\lambda_1\), the actor balancing coefficient \(\lambda_2\), and the learned preference planner. Specifically, \(\lambda_1\) controls the strength of the pairwise diversity regularizer in the planner objective, \(\lambda_2\) balances the team advantage and the scalarized individual advantage in the actor update. Moreover, to validate the effectiveness of our preference coordination, we compare PCMA with RAND and SAME. RAND samples each agent's preference randomly without a learned planner, while SAME forces all agents to use the same preference by averaging their sampled preferences.

\begin{figure*}[!htb]
    \centering
    \begin{subfigure}[t]{0.32\textwidth}
        \centering
        \includegraphics[width=\linewidth]{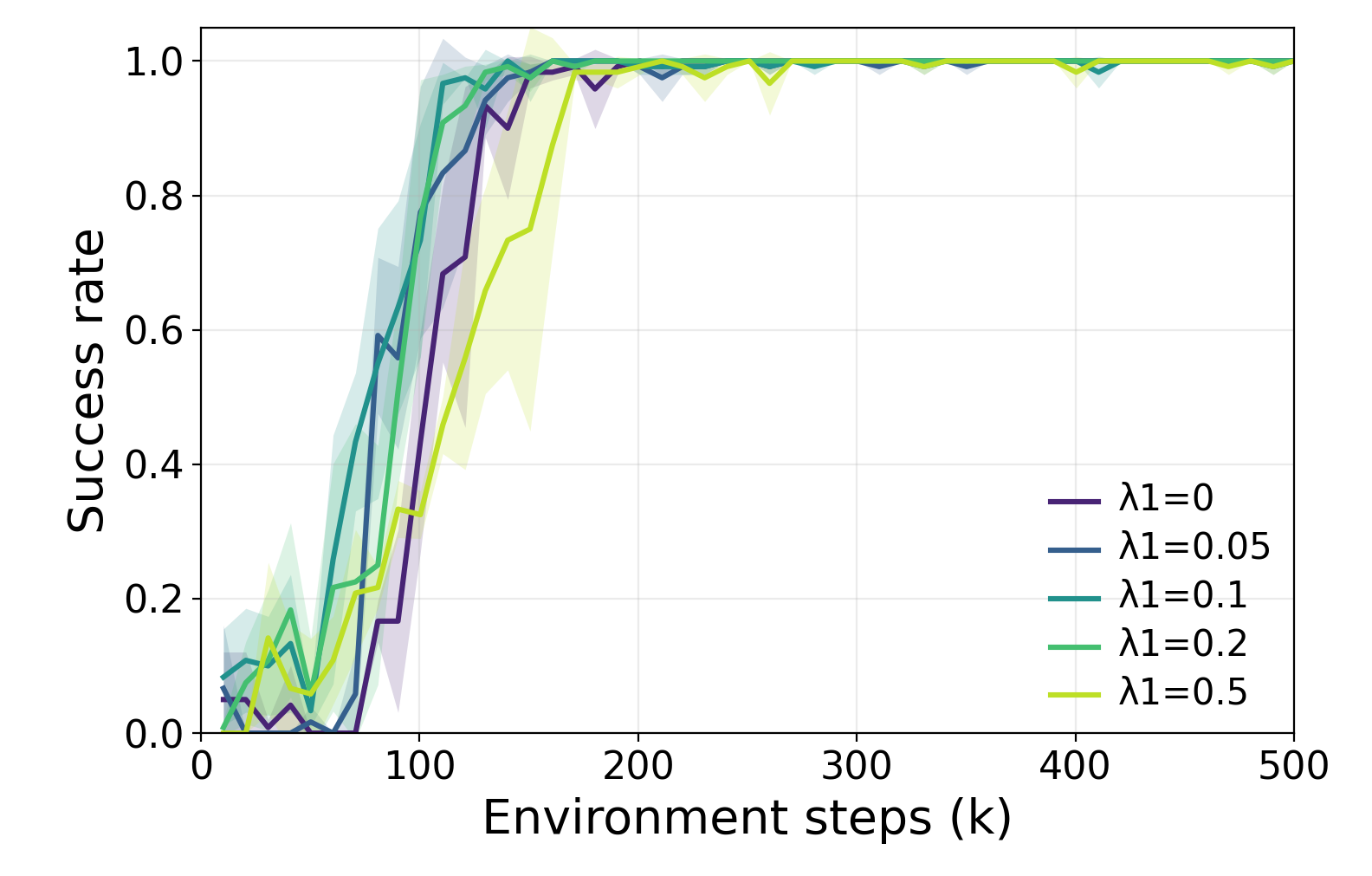}
        \caption{Study of $\lambda_1$.}
        \label{fig:ablation_lambda1}
    \end{subfigure}
    \hfill
    \begin{subfigure}[t]{0.32\textwidth}
        \centering
        \includegraphics[width=\linewidth]{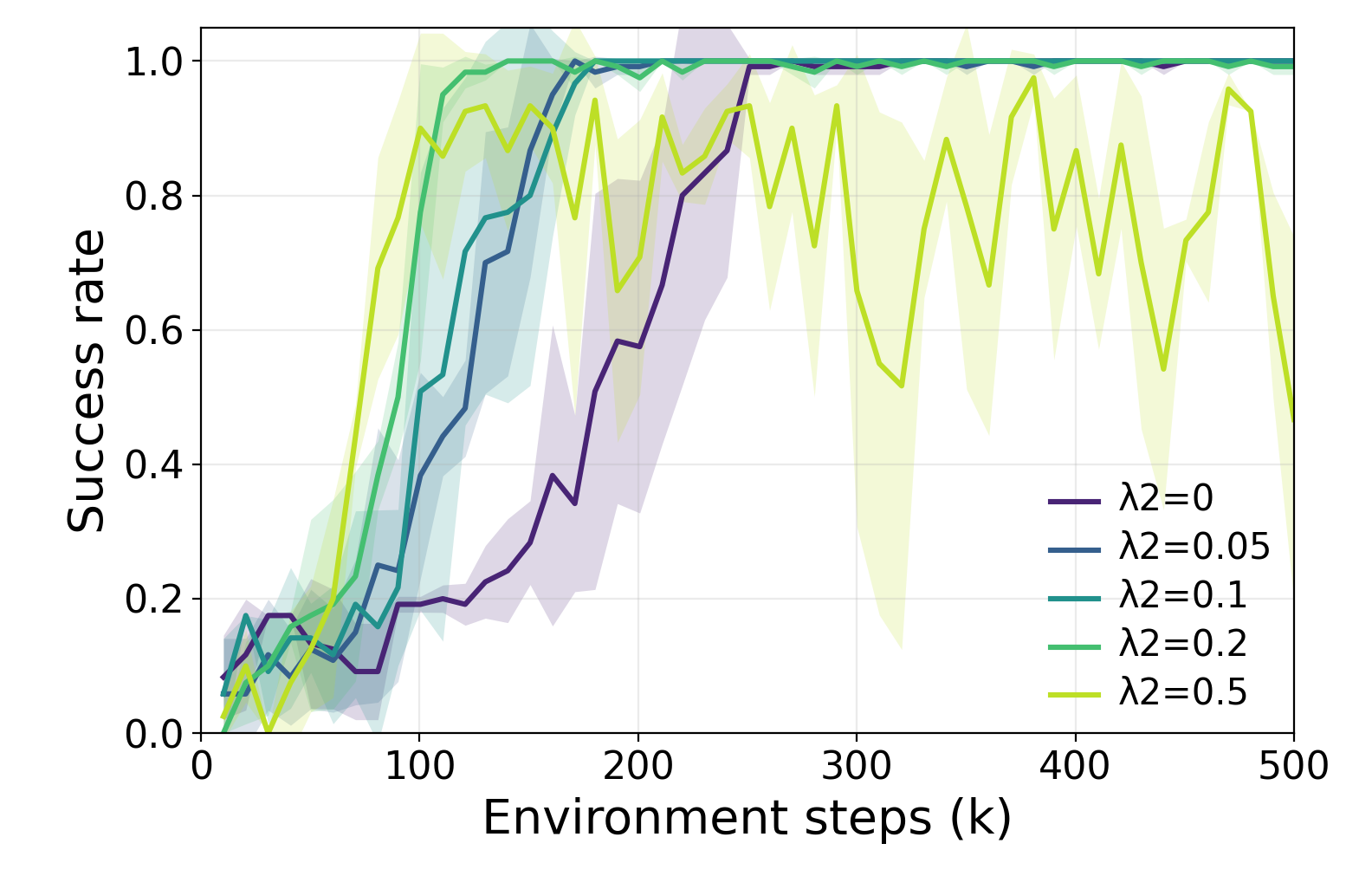}
        \caption{Study of $\lambda_2$.}
        \label{fig:ablation_lambda2}
    \end{subfigure}
    \hfill
    \begin{subfigure}[t]{0.32\textwidth}
        \centering
        \includegraphics[width=\linewidth]{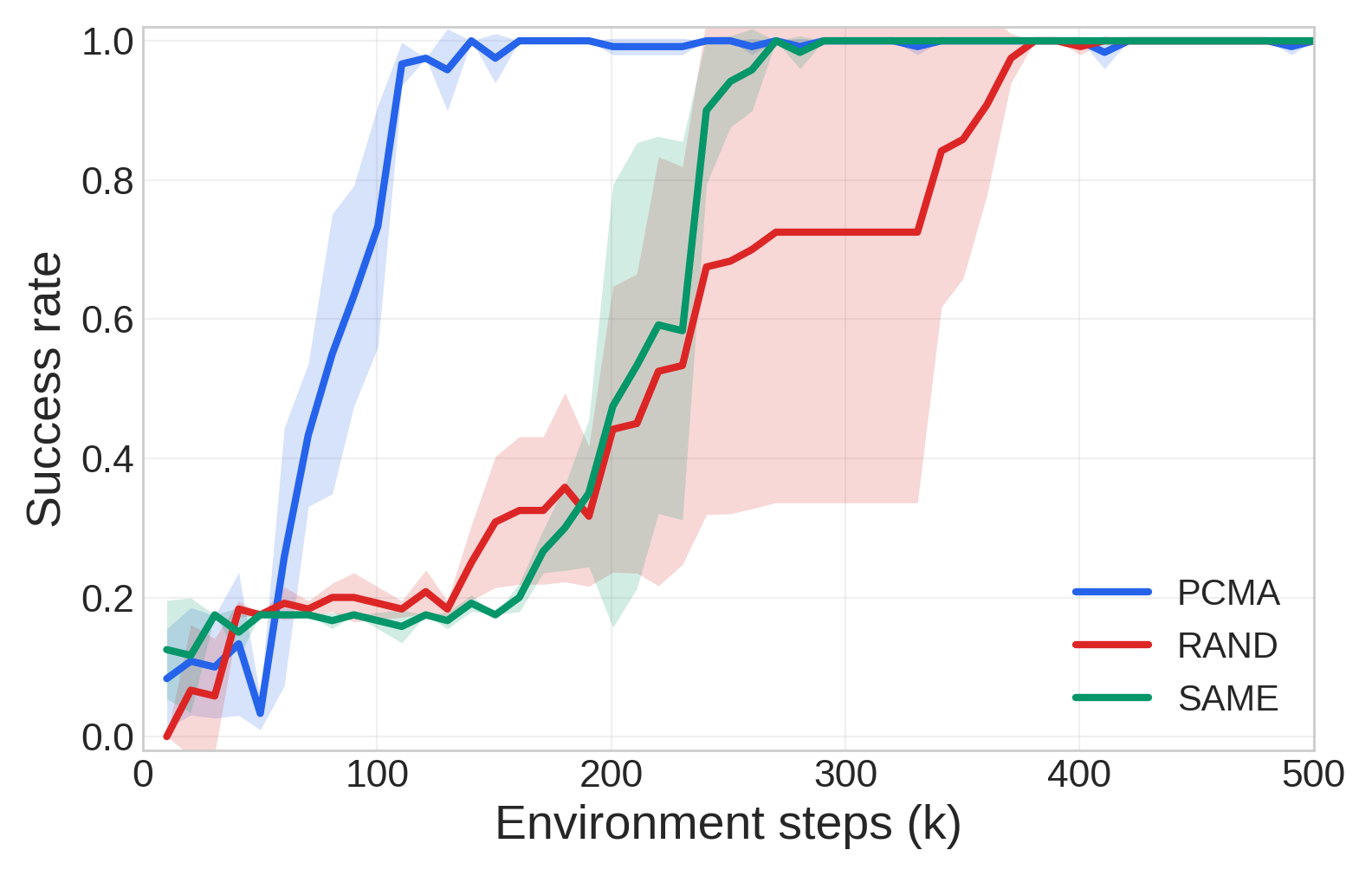}
        \caption{Study of preference planner.}
        \label{fig:ablation_planner}
    \end{subfigure}

    \caption{
    Ablation studies on Cooperative Spread.
    (a) $\lambda_1=0.1$ and $\lambda_1=0.2$ improve early learning over $\lambda_1=0$, suggesting that preference diversity helps avoid collapse, while $\lambda_1=0.5$ slows learning due to excessive diversity.
    (b) $\lambda_2=0.1$ and $\lambda_2=0.2$ achieve stable convergence, whereas $\lambda_2=0.5$ becomes unstable, showing that individual guidance should not dominate the team signal.
    (c) The planner learns coordinated, agent-specific preferences rather than merely injecting random diversity.
    }
    \label{fig:ablation}
\end{figure*}



\section{Conclusion and Future Work}
\label{sec:conclusion}

We model cooperative MOMARL as a mo-POSG with partially aligned interests and seek a team-optimal equilibrium. Our analysis links preference diversity to first-order team improvement and establishes local equilibrium tracking. Based on this analysis, we propose PCMA to coordinate agent-specific preferences. Experiments show that PCMA can improve team performance and induce interpretable agent specialization.

One limitation of this work is that our study focuses on cooperative benchmarks where the objectives and reward decomposition are explicitly specified. A promising direction for future work is to extend preference coordination to more complex real-world multi-agent systems, such as autonomous driving and open-ended agentic systems, where agents may have different capabilities, contexts, and objectives that naturally emerge from real task demands.



\subsection*{AI use statement}
Generative AI tools were used to polish language and improve manuscript clarity, assist with algorithm implementation, check citations and bibliographic metadata, and resolve LaTeX formatting and layout issues. AI assistance was also used to refine theoretical formulations and mathematical statements and clarify the interpretation of experimental results. The authors take responsibility for the final content, including AI-assisted text, code, mathematical claims, and references.

\subsection*{Reproducibility statement}
Section~\ref{sec:theoretical_analysis} states the theoretical assumptions and results, with proofs in Appendix~\ref{sec:appendix_proofs}.
Section~\ref{sec:methods} specifies the policy, critic, and preference-planner objectives, while Algorithm~\ref{alg:pcma_detailed} details rollout collection, advantage estimation, and parameter updates.
Appendix~\ref{subsec:env_details} documents the task setups, action spaces, reward definitions, and success criteria.
Appendix~\ref{app:opencda_details} describes the OpenCDA-MARL/CARLA setup and evaluation protocol.
Appendix~\ref{subsec:training_details} provides network architectures, training hyperparameters, computing hardware, and measured training times.

\newpage
\appendix
\section{Detailed Proofs for Theoretical Results}
\label{sec:appendix_proofs}
\subsection{Assumptions}
We begin by restating the assumptions required for the theoretical analysis.
\begin{assumption}[Preference-improvement alignment]
\label{ass_re:preference_improvement_alignment}
There exists \(\kappa>0\) such that, for every agent \(i\) with \(\tilde p_i\neq 0\), the projection of \(\tilde b_i\) onto \(\tilde p_i\) is lower bounded by \(\kappa\), i.e.,
\[
\frac{\tilde p_i^\top \tilde b_i}{\|\tilde p_i\|_2^2}
\ge
\kappa,
\qquad i=1,\ldots,N .
\]
where \(\tilde p_i:=p_i-\bar p\) and \(\tilde b_i:=b_i-\bar b\).
\end{assumption}

\begin{assumption}[Attraction property near equilibrium]
\label{ass_re:attraction}
Let \(\mathbf\theta(\mathbf p)\) denote a strict local Nash equilibrium under the preference profile \(\mathbf p\). Assume that there exist \(\rho\in(0,1)\) and a neighborhood of \(\mathbf\theta(\mathbf p)\) such that
\[
\|\mathbf\theta_{\mathrm{new}}-\mathbf\theta(\mathbf p)\|
\le
\rho
\|\mathbf\theta_{\mathrm{old}}-\mathbf\theta(\mathbf p)\|,
\]
for all \(\mathbf\theta_{\mathrm{old}}\) in this neighborhood.
\end{assumption}

\begin{assumption}[Local Lipschitz continuity]
\label{ass_re:lipschitz}
Let \(\mathbf\theta(\cdot):\mathcal U\to\Theta\) be the strict local Nash equilibrium mapping induced by Lemma~\ref{lem:local_continuity_stationary_solution}. Assume that this mapping is Lipschitz continuous in \(\mathcal U\): there
exists \(C>0\) such that for any two preference profiles
\(\mathbf p,\mathbf p'\in\mathcal U\),
\[
\|\mathbf\theta(\mathbf p')-\mathbf\theta(\mathbf p)\|
\le
C\|\mathbf p'-\mathbf p\|.
\]
\end{assumption}

\subsection{Proof of Theorem 4.2 in the main text}
\begin{theorem}[Team Improvement Decomposition]
\label{theo_re:team_improvement_decomposition}
Under Assumption~\ref{ass_re:preference_improvement_alignment}, the team improvement as $\eta\to0^+$ satisfies
\[
J_{\mathrm{team}}(\mathbf\theta_{\mathrm{new}})
-
J_{\mathrm{team}}(\mathbf\theta)
\ge
\eta
\sum_{i=1}^{N}
\left\|
\nabla_{\theta_i}J_{\mathrm{team}}(\mathbf\theta)
\right\|_2^2
+
\eta N
\left(
\bar p^\top \bar b
+
\kappa \mathcal D_p
\right)
+o(\eta)
\]
where
\[
\mathcal D_p
:=
\frac{1}{2N^2}
\sum_{i=1}^{N}
\sum_{j=1}^{N}
\|p_i-p_j\|_2^2 .
\]
is the pairwise preference distance across agents.
\end{theorem}
\begin{proof}
Given preferences \(p_1,\ldots,p_N\), each agent takes one gradient update
\[
\theta_{i,\mathrm{new}}
=
\theta_i
+
\eta
\nabla_{\theta_i}
U_i(\theta;p_i),
\qquad i=1,\ldots,N .
\]
where
\[
U_i(\theta;p_i)
=
J_{\mathrm{team}}(\theta)
+
\langle p_i,J_i(\theta)\rangle,
\]
We write
\(\theta_{\mathrm{new}}=(\theta_{1,\mathrm{new}},\ldots,\theta_{N,\mathrm{new}})\).
Let $\Delta J:=J_{\mathrm{team}}(\theta_{\mathrm{new}})-J_{\mathrm{team}}(\theta)$.
By the first-order Taylor expansion of \(J_{\mathrm{team}}\) at \(\theta\),
\begin{align*}
\Delta J
&=
\eta
\sum_{i=1}^{N}
\left\langle
\nabla_{\theta_i}J_{\mathrm{team}}(\theta),
\nabla_{\theta_i}U_i(\theta;p_i)
\right\rangle
+
o(\eta)
\\
&=
\eta
\sum_{i=1}^{N}
\left\langle
\nabla_{\theta_i}J_{\mathrm{team}}(\theta),
\nabla_{\theta_i}J_{\mathrm{team}}(\theta)
+
\nabla_{\theta_i}\langle p_i,J_i(\theta)\rangle
\right\rangle
+
o(\eta)
\\
&=
\eta
\sum_{i=1}^{N}
\left\|
\nabla_{\theta_i}J_{\mathrm{team}}(\theta)
\right\|_2^2
+
\eta
\sum_{i=1}^{N}
\left\langle
\nabla_{\theta_i}J_{\mathrm{team}}(\theta),
\sum_{k=1}^{K}p_{i,k}\nabla_{\theta_i}J_{i,k}(\theta)
\right\rangle
+
o(\eta)
\\
&=
\eta
\sum_{i=1}^{N}
\left\|
\nabla_{\theta_i}J_{\mathrm{team}}(\theta)
\right\|_2^2
+
\eta
\sum_{i=1}^{N}
\sum_{k=1}^{K}
p_{i,k}
\left\langle
\nabla_{\theta_i}J_{\mathrm{team}}(\theta),
\nabla_{\theta_i}J_{i,k}(\theta)
\right\rangle
+
o(\eta)
\\
&=
\eta
\sum_{i=1}^{N}
\left\|
\nabla_{\theta_i}J_{\mathrm{team}}(\theta)
\right\|_2^2
+
\eta
\sum_{i=1}^{N}
\langle p_i,b_i\rangle
+
o(\eta),
\end{align*}

Let \(\tilde p_i=p_i-\bar p\) and \(\tilde b_i=b_i-\bar b\). Then \(\sum_{i=1}^{N}\tilde p_i=0\) and \(\sum_{i=1}^{N}\tilde b_i=0\). Let \(P=(\tilde p_1,\ldots,\tilde p_N)^\top\in\mathbb R^{N\times K}\) be the centered preference matrix, and then crossing term can be decomposed as follows.

\begin{align}
\Delta J
&=
\eta
\sum_{i=1}^{N}
\left\|
\nabla_{\theta_i}J_{\mathrm{team}}(\theta)
\right\|_2^2
+
\eta
\sum_{i=1}^{N}
\left\langle
\bar p+\tilde p_i,
\bar b+\tilde b_i
\right\rangle
+
o(\eta)
\notag
\\
&=
\eta
\sum_{i=1}^{N}
\left\|
\nabla_{\theta_i}J_{\mathrm{team}}(\theta)
\right\|_2^2
+
\eta
\left(
N\langle \bar p,\bar b\rangle
+
\sum_{i=1}^{N}
\langle \tilde p_i,\tilde b_i\rangle
\right)
+
o(\eta)
\notag
\\
&\ge
\eta
\sum_{i=1}^{N}
\left\|
\nabla_{\theta_i}J_{\mathrm{team}}(\theta)
\right\|_2^2
+
\eta
\left(
N\langle \bar p,\bar b\rangle
+
\kappa
\sum_{i=1}^{N}
\|\tilde p_i\|_2^2
\right)
+
o(\eta)
\notag\\
&=
\eta
\sum_{i=1}^{N}
\left\|
\nabla_{\theta_i}J_{\mathrm{team}}(\theta)
\right\|_2^2
+
\eta N \bar p^\top\bar b
+
\eta  \kappa \|P\|_F^2 
+
o(\eta)
\notag
\label{eq:team_improvement_lower_bound}
\end{align}
The inequality uses Assumption~\ref{ass_re:preference_improvement_alignment},
and $\|\cdot\|_F$ denotes the Frobenius norm. Finally, the squared Frobenius
norm of the centered preference matrix, divided by $N$, equals the pairwise
preference distance:
\[
\begin{aligned}
\frac{1}{N}\|P\|_F^2
&=
\frac{1}{N}
\sum_{i=1}^{N}
\|p_i-\bar p\|_2^2
\\
&=
\frac{1}{N}
\sum_{i=1}^{N}
\|p_i\|_2^2
-
\|\bar p\|_2^2
\quad \text{(by the variance identity)}
\\
&=
\frac{1}{N}
\sum_{i=1}^{N}
\|p_i\|_2^2
-
\frac{1}{N^2}
\sum_{i=1}^{N}
\sum_{j=1}^{N}
p_i^\top p_j
\quad \text{(by bilinearity)}
\\
&=
\frac{1}{2N^2}
\sum_{i=1}^{N}
\sum_{j=1}^{N}
\left(
\|p_i\|_2^2
+
\|p_j\|_2^2
-
2p_i^\top p_j
\right)
\\
&=
\frac{1}{2N^2}
\sum_{i=1}^{N}
\sum_{j=1}^{N}
\|p_i-p_j\|_2^2
=
\mathcal D_p .
\end{aligned}
\]
We obtain the following asymptotic lower bound:
\[
J_{\mathrm{team}}(\theta_{\mathrm{new}})
-
J_{\mathrm{team}}(\theta)
\ge
\eta
\sum_{i=1}^{N}
\left\|
\nabla_{\theta_i}J_{\mathrm{team}}(\theta)
\right\|_2^2
+
\eta N \bar p^\top\bar b
+
\eta  \kappa  N \mathcal D_p
+o(\eta).
\]
\end{proof}

\subsection{Proof of Lemma 4.3 in the main text} 
\label{pf:continual_equilibrium}
\begin{lemma}[Continuity of preference-conditioned local Nash equilibria]
\label{lem_re:local_continuity_stationary_solution}
Suppose that \(\mathbf\theta^{*}\in\operatorname{int}\Theta\) is a strict local Nash equilibrium of
\(\mathcal G(\mathbf p)\) satisfying
\[
\nabla_{\theta_i}U_i(\mathbf\theta^{*};p_i)=0,
\qquad
\nabla^2_{\theta_i\theta_i}U_i(\mathbf\theta^{*};p_i)\prec0,
\qquad i=1,\ldots,N .
\]
Assume that the joint gradient
\((\nabla_{\theta_1}U_1,\ldots,\nabla_{\theta_N}U_N)\) is continuously differentiable near \((\mathbf\theta^{*},\mathbf p)\), and its Jacobian with respect to \(\mathbf\theta\) is nonsingular at this point. Then there exist neighborhoods \(\mathcal U\) of \(\mathbf p\) and \(\mathcal V\subset\Theta\) of \(\mathbf\theta^{*}\), and a unique continuously differentiable mapping \(\mathbf\theta(\cdot):\mathcal U\to\mathcal V\) with \(\mathbf\theta(\mathbf p)=\mathbf\theta^{*}\), such that, for every \(\tilde{\mathbf p}\in\mathcal U\), \(\mathbf\theta(\tilde{\mathbf p})\) is the unique stationary point in \(\mathcal V\) and a strict local Nash equilibrium of \(\mathcal G(\tilde{\mathbf p})\).
\end{lemma}
\begin{proof}
Define the joint stationarity mapping
\[
F(\mathbf\theta,\mathbf p)
:=
\bigl(
\nabla_{\theta_1}U_1(\mathbf\theta;p_1),
\ldots,
\nabla_{\theta_N}U_N(\mathbf\theta;p_N)
\bigr).
\]
The assumed first-order conditions give
\[
F(\mathbf\theta^{*},\mathbf p)=0 .
\]
By assumption, \(F\) is continuously differentiable near
\((\mathbf\theta^{*},\mathbf p)\), and its Jacobian \(\nabla_{\mathbf\theta}F(\mathbf\theta^{*},\mathbf p)
\) is nonsingular. Therefore, by the Implicit Function Theorem, there exist
sufficiently small neighborhoods \(\mathcal U\) of \(\mathbf p\) and \(\mathcal V\subset\operatorname{int}\Theta\) of \(\mathbf\theta^{*}\), and a unique continuously
differentiable mapping \(\mathbf\theta(\cdot):\mathcal U\to\mathcal V\) such that
\(\mathbf\theta(\mathbf p)=\mathbf\theta^{*}\) and
\[
F(\mathbf\theta(\tilde{\mathbf p}),\tilde{\mathbf p})=0,
\qquad
\forall\,\tilde{\mathbf p}\in\mathcal U .
\]
For each \(\tilde{\mathbf p}\in\mathcal U\), \(\mathbf\theta(\tilde{\mathbf p})\) is the unique stationary point in \(\mathcal V\), satisfying
\[
\nabla_{\theta_i}
U_i(\mathbf\theta(\tilde{\mathbf p});\tilde p_i)=0,
\qquad i=1,\ldots,N .
\]
The Hessians are negative definite at \((\mathbf\theta^{*},\mathbf p)\).
By continuity of the Hessians and \(\mathbf\theta(\cdot)\), this property also holds for all sufficiently nearby preference profiles:
\[
\nabla^2_{\theta_i\theta_i}
U_i(\mathbf\theta(\tilde{\mathbf p});\tilde p_i)\prec0,
\qquad i=1,\ldots,N,\quad \tilde{\mathbf p}\in\mathcal U .
\]
For each agent, hold the other agents' policies fixed. Stationarity and the second-order conditions then imply that its policy is a strict local maximizer of its payoff. Hence every \(\mathbf\theta(\tilde{\mathbf p})\) is a strict local Nash equilibrium of \(\mathcal G(\tilde{\mathbf p})\), which proves the result.
\end{proof}

\subsection{Proof of Theorem 4.6 in the main text}
\label{pf_tracking}
\begin{theorem}[Equilibrium tracking]
\label{thm_re:tracking}
Consider the iterative policy update
\[
\theta_i^{t+1}
=
\theta_i^t
+
\eta
\nabla_{\theta_i}U_i(\theta^t;\mathbf p^t),
\qquad i=1,\ldots,N .
\]
Under Assumption 4.4 
and Assumption 4.5 on page 5 of the main text, 
the tracking error
$
e_t
:=
\|\theta^t-\theta(\mathbf p^t)\|
$
satisfies
\[
e_{t+1}
\le
\rho e_t
+
C\|\mathbf p^{t+1}-\mathbf p^t\|.
\]
In particular, if
\(\|\mathbf p^{t+1}-\mathbf p^t\|\le\delta\), then
\[
\limsup_{t\to\infty} e_t
\le
\frac{C}{1-\rho}\delta .
\]
\end{theorem}

\begin{proof}
For any \(t\ge 0\), by adding and subtracting \(\theta(\mathbf p^t)\), we have
\[
\begin{aligned}
e_{t+1}
&=
\|\theta^{t+1}-\theta(\mathbf p^{t+1})\|
\\
&\le
\|\theta^{t+1}-\theta(\mathbf p^t)\|
+
\|\theta(\mathbf p^t)-\theta(\mathbf p^{t+1})\|.
\end{aligned}
\]

First, by Assumption 4.4,
the policy update is locally contractive around
\(\theta(\mathbf p^t)\). Therefore,
\[
\|\theta^{t+1}-\theta(\mathbf p^t)\|
\le
\rho
\|\theta^t-\theta(\mathbf p^t)\|
=
\rho e_t .
\]
Second, by Assumption 4.5, 
the equilibrium path changes Lipschitz continuously with the preference profile, so
\[
\|\theta(\mathbf p^t)-\theta(\mathbf p^{t+1})\|
\le
C\|\mathbf p^{t+1}-\mathbf p^t\|.
\]
Combining the two bounds gives
\[
e_{t+1}
\le
\rho e_t
+
C\|\mathbf p^{t+1}-\mathbf p^t\|.
\]

If \(\|\mathbf p^{t+1}-\mathbf p^t\|\le\delta\), then
\[
e_{t+1}
\le
\rho e_t + C\delta .
\]
Applying this inequality recursively gives
\[
\begin{aligned}
e_t
&\le
\rho e_{t-1}+C\delta
\\
&\le
\rho^2 e_{t-2}
+
C\delta(1+\rho)
\\
&\le
\cdots
\\
&\le
\rho^t e_0
+
C\delta
\sum_{k=0}^{t-1}\rho^k\\
&=
\rho^t e_0
+
C\delta
\frac{1-\rho^t}{1-\rho}.
\end{aligned}
\]
Since $\rho\in(0,1)$, take \(\limsup_{t\to\infty}\), and we obtain
\[
\limsup_{t\to\infty} e_t
\le
\frac{C}{1-\rho}\delta .
\]
\end{proof}

\section{Implementation Details}
To support reproducibility, we provide detailed environment specifications, training hyperparameters, and pseudocode in the appendix. The source code, including environment modifications and training scripts, will be released upon acceptance.
\label{sec:imple_details}
\subsection{Environment Setting}
\label{subsec:env_details}
This subsection provides additional details of the environments used in our experiments. We describe the task setup, action space, reward design, and success criterion for each environment. The detailed reward design and success criterion are summarized in Table~\ref{tab:env_settings}.
\begin{figure}[t]
\centering
\begin{subfigure}[t]{0.24\linewidth}
    \centering
    \includegraphics[width=\linewidth]{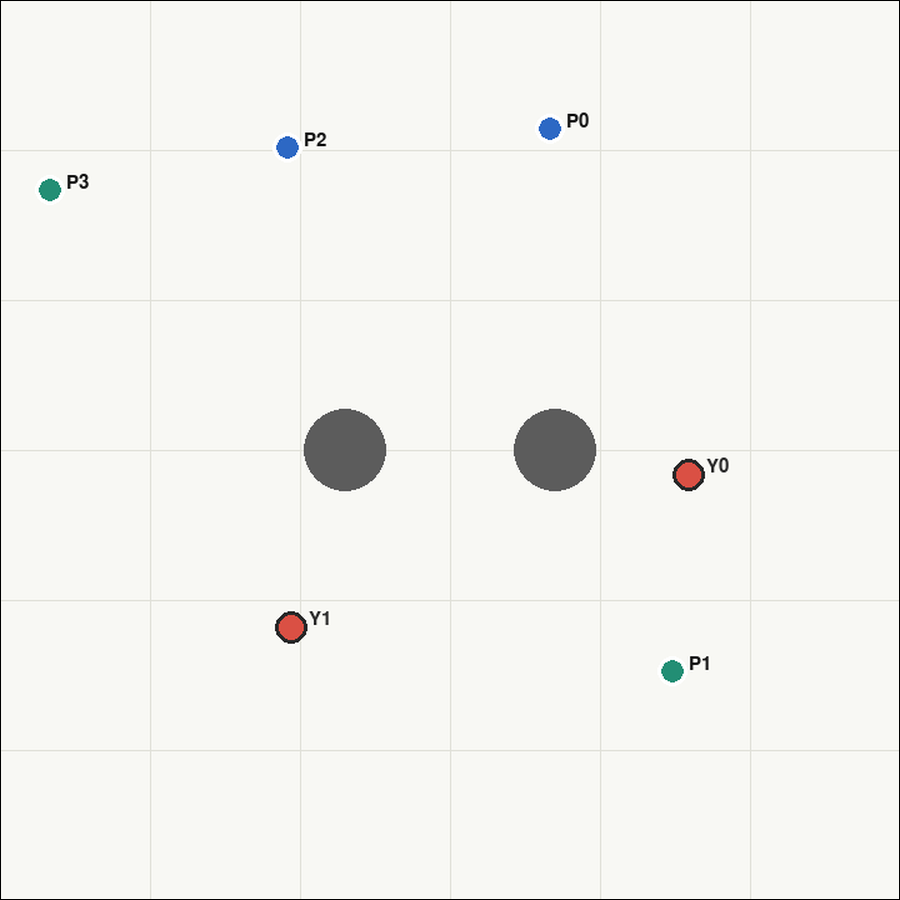}
    \caption{MOMPE}
\end{subfigure}
\hfill
\begin{subfigure}[t]{0.24\linewidth}
    \centering
    \includegraphics[width=\linewidth]{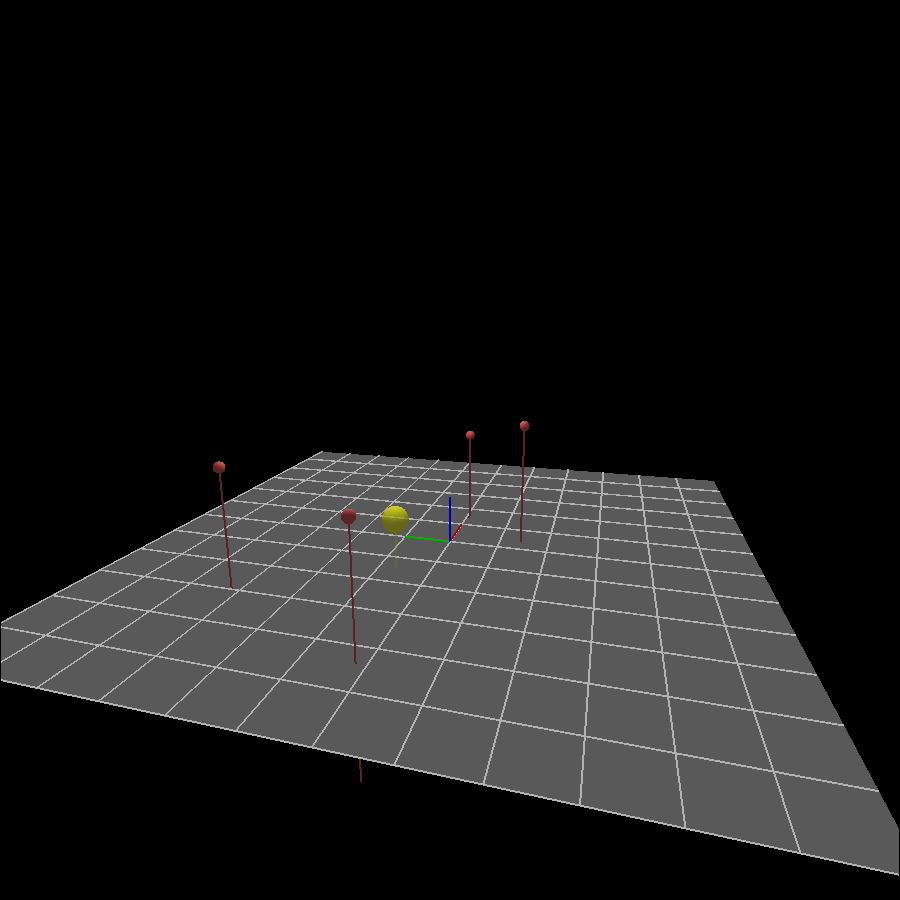}
    \caption{Catch/Escort}
\end{subfigure}
\hfill
\begin{subfigure}[t]{0.24\linewidth}
    \centering
    \includegraphics[width=\linewidth]{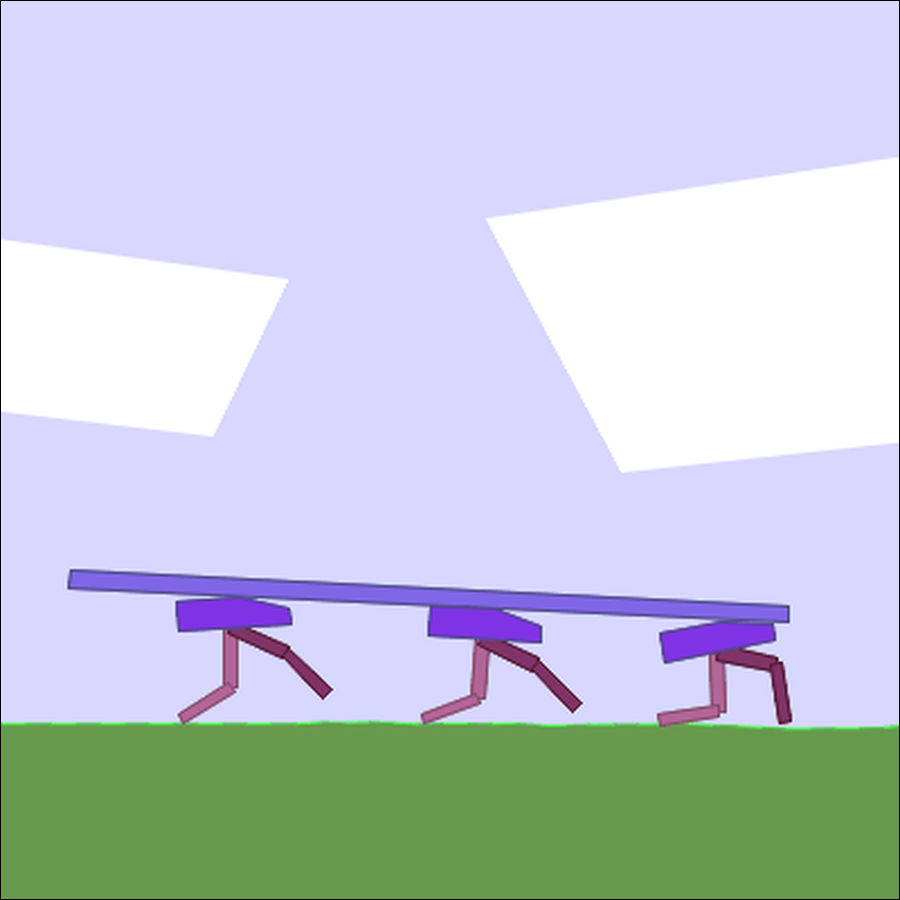}
    \caption{MOMAWalker}
\end{subfigure}
\hfill
\begin{subfigure}[t]{0.24\linewidth}
    \centering
    \includegraphics[width=\linewidth]{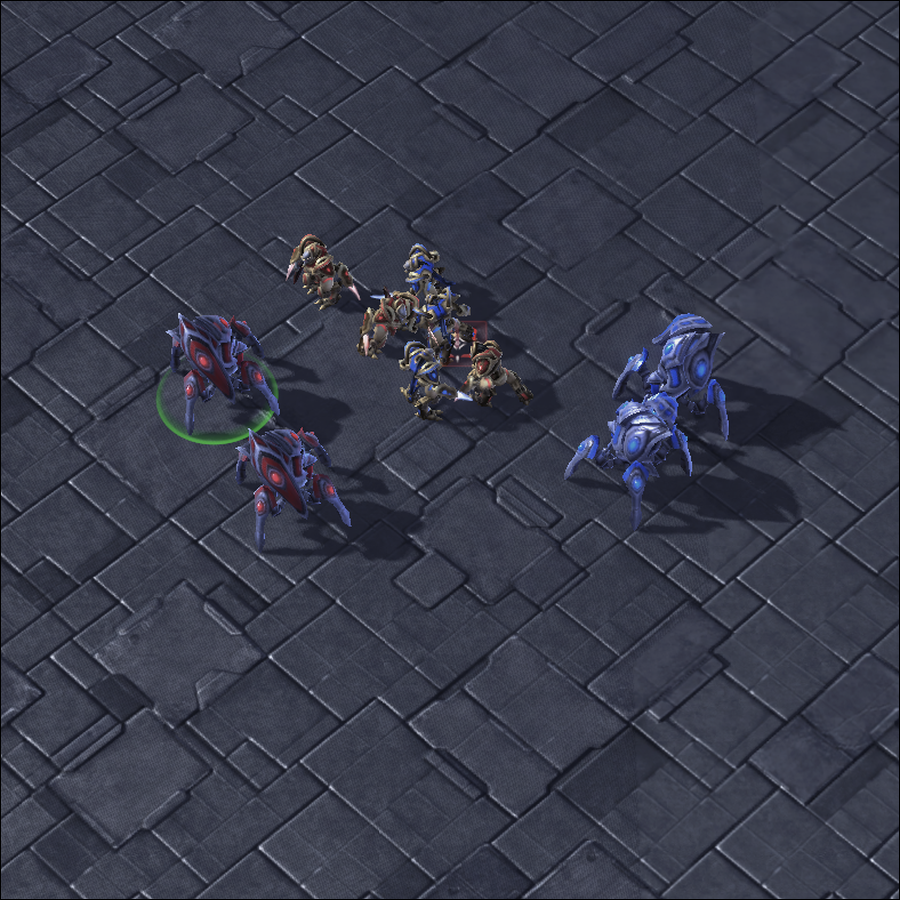}
    \caption{SMAC}
\end{subfigure}
\caption{Environments used in our experiments.}
\label{fig:env-snapshots}
\end{figure}

\begin{table}[t]
\centering
\caption{Reward and success details}
\label{tab:env_settings}
\small
\setlength{\tabcolsep}{3pt}
\renewcommand{\arraystretch}{1.3}
\begin{tabular}{llp{3.0cm}p{3.0cm}p{3.2cm}}
\toprule
Environment & Dimension & Team reward & Individual reward & Success criterion \\
\midrule
Cooperative Spread & 2A2O &
Sparse task reward or punishment: $+1$ when both landmark reached, $-1$ on collision.&
[- distance to landmark 1, - distance to landmark 2]. &
Two agents occupy two distinct landmarks. \\
\midrule
Safe Predator-Prey & 4A3O &
Sparse task reward or punishment: $+1$ when catch prey, $-1$ on collision or bump into obstacle.&
[progress to prey 1, progress to prey 2, -close distance to obstacle/other predator]. &
All prey are captured. \\
\midrule
Catch & 4A2O &
Sparse task reward or punishment: $+1$ when catch target, $-1$ on collision. &
[progress to target, -close distance to other drone] &
catch the target \\
\midrule
Escort & 4A2O &
only collision penalty &
[- distance to target, distance to other drones] &
Not Applicable \\
\midrule
MOMAWalker & 3A2O &
Package progress distance, fall/drop penalties. &
[Walker progress, stability penalty] &
Not Applicable \\
\midrule
SMAC & xA2O &
+1 on success. &
[damage dealt, -health loss] &
Defeat all enemies within episode time \\
\midrule
OpenCDA-MARL & $\leq$10A2O &
Route success, progress, step penalty, and collision penalty. &
[efficiency, safety/\newline interaction risk] &
Route completion without collision. \\
\bottomrule
\end{tabular}
\end{table}
\paragraph{Cooperative Spread and Safe Predator-Prey.}
We use two particle-world tasks with discrete actions. Spread is a two-agent, two-landmark variant of MPE simple spread, where agents must coordinate to cover distinct landmarks. Predator-Prey uses four predators, two moving prey, and obstacles in a bounded 2D arena. Compared with Spread, Predator-Prey introduces moving targets, obstacle avoidance, and a three-objective reward vector that separates progress toward the two prey from safety.

\paragraph{Catch and Escort.}
Catch and Escort are continuous-control drone tasks adapted from CrazyRL/MOMALand. In both environments, four drones move in a bounded 3D space with 3D velocity actions. Catch focuses on approaching and capturing a target while avoiding unsafe proximity to other drones. Escort instead requires drones to maintain an escort formation around a moving target as it travels toward a goal. Thus, Catch emphasizes interception, while Escort emphasizes sustained formation control.

\paragraph{MOMAWalker.}
MOMAWalker is adapted from the Multiwalker domain. Three walkers jointly transport a package across terrain using continuous joint-torque actions. We use this environment to evaluate coordinated continuous control under a shared physical task, where agents must trade off forward progress with package and body stability.

\paragraph{SMAC.}
SMAC comprises cooperative StarCraft Multi-Agent Challenge maps, including 3m, 2s3z, and 8m, and the difficulty is set to 5. In these tasks, agents control decentralized units using discrete actions with action masks. We use a two-dimensional combat reward vector based on damage dealt and health loss, together with a sparse team-success signal indicating whether all enemy units are defeated within the episode limit.

\paragraph{OpenCDA-MARL.}
OpenCDA-MARL extends OpenCDA with a MARL controller for CARLA intersection control. The environment is an unsignalized four-way intersection on a local OpenDRIVE map, with traffic entering from north, south, east, and west, and each episode lasts at most $2400$ simulation steps. During training, we use live traffic with at most $10$ simultaneously active vehicles, all of them MARL-controlled CAVs (there is no human-driven background traffic); for checkpoint evaluation, we replay a fixed traffic file so that all algorithms face the same arrival pattern.

\subsection{OpenCDA-MARL/CARLA Experimental Details}
\label{app:opencda_details}
CAVs spawn on the four approaches and drive through the intersection to their destinations. Each CAV follows a fixed route generated by the OpenCDA planner. The learned action is a one-dimensional target speed command in $\mathrm{km \cdot h^{-1}}$, clipped to $[0,45]$, and the low-level controller tracks this speed along the assigned route. The learned speed is applied directly: CAVs have no rule-based right-of-way or collision-avoidance layer, so the collision rate measures the learned policies themselves, which explains the high absolute collision rates of all methods. The experiment therefore tests learned yielding and crossing-order behavior under speed control, not route selection. The base observation includes relative position to the intersection, absolute position, lane position, heading angle, distance to the intersection, distance to the front vehicle, and waypoint-buffer information. Preference-conditioned variants append a two-dimensional preference vector to this state.

\paragraph{PCMA-SAC.}
In this benchmark, PCMA is instantiated on a SAC backbone. Each CAV is assigned a preference vector when it spawns; the vector is appended to its state and weights its two objectives. During training, a learned, observation-dependent shift adapts the preference that conditions each CAV's actor. The Dirichlet preference planner of Section~\ref{sec:methods} is not used here, so Table~\ref{tab:opencda_carla_results} validates preference-conditioned coordination in the off-policy setting rather than the full on-policy algorithm.

\paragraph{Rewards and objectives.}
We use two reward profiles, \texttt{cooperative} and \texttt{competitive}, with the same learned terms: collision penalty $-400$, success reward $300$, step penalty $-0.8$, and progress shaping scaled by $0.4$. Both group the reward components into two objectives: efficiency (progress, success, and step penalty) and a collision-driven second objective, called safety in the cooperative profile and interaction risk in the competitive one. Both settings also share the traffic and the preference assignment: CAVs entering from the north and east receive the preference $[0.6, 0.4]$ over efficiency and safety, and those entering from the south and west receive $[0.4, 0.6]$. The competitive setting differs in two respects: it additionally records time-to-collision, yielding, clearance, and per-group outcomes as metric-only signals, and PCMA-SAC additionally uses a separation term that encourages the adapted preferences of the two groups to stay at least $0.25$ apart (weight $0.02$). SAC and MAPPO use the same configuration in both settings, so their two rows are independent repetitions.

\paragraph{Metrics and protocol.}
We report success rate (percentage of vehicles that complete their route), collision rate (percentage of vehicles that collide), and throughput (successful completions per simulated hour). Each run is evaluated at the checkpoint of its highest-return training episode, on the fixed evaluation replay; we report the mean and standard deviation over three seeds.

\subsection{OpenCDA-MARL/CARLA Results}
\label{subsec:opencda_marl_carla_validation}
\label{app:opencda_results}

We evaluate PCMA-SAC on OpenCDA-MARL~\citep{OpenCDAMARL} in CARLA~\citep{CARLA}; Table~\ref{tab:opencda_carla_results} reports the results under the protocol in Appendix~\ref{app:opencda_details}.

\begin{table*}[!ht]
\centering
\caption{OpenCDA-MARL/CARLA best-checkpoint evaluation. Mean $\pm$ std over three seeds; bold marks the best mean per setting. Without preference separation, PCMA-SAC performs comparably to SAC; with preference separation, it favors a more conservative efficiency--safety trade-off.}
\label{tab:opencda_carla_results}
\small
\setlength{\tabcolsep}{5pt}
\renewcommand{\arraystretch}{1.08}
\begin{tabular}{llccc}
\toprule
Setting & Method
& Success (\%) $\uparrow$
& Collision (\%) $\downarrow$
& Throughput (veh/h) $\uparrow$ \\
\midrule

Coop. & SAC      & 68.6 $\pm$ 1.9 & 31.4 $\pm$ 1.9 & 1692 $\pm$ 46 \\
Coop. & MAPPO    & 55.4 $\pm$ 4.7 & 42.5 $\pm$ 3.5 & 1417 $\pm$ 100 \\
Coop. & PCMA-SAC & \textbf{69.6 $\pm$ 1.2} & \textbf{30.4 $\pm$ 1.2} & \textbf{1716 $\pm$ 29} \\


\midrule
Comp. & SAC      & 67.9 $\pm$ 0.7 & 32.1 $\pm$ 0.7 & \textbf{1674 $\pm$ 17} \\
Comp. & MAPPO    & 42.1 $\pm$ 30.3 & \textbf{23.3 $\pm$ 17.6} & 1066 $\pm$ 759 \\
Comp. & PCMA-SAC & \textbf{68.0 $\pm$ 2.2} & \textbf{23.3 $\pm$ 8.3} & 1192 $\pm$ 291 \\
\bottomrule
\end{tabular}
\end{table*}

\subsection{Training Details}
\label{subsec:training_details}
Algorithm~\ref{alg:pcma_detailed} gives the full training procedure of PCMA.

\paragraph{Training setup.}
Unless otherwise stated, all on-policy methods use the same PPO-style training setup. We use Adam for both actor and critic optimization, and all MLPs use Tanh
activations. For PCMA and MAPPO-style baselines, the actor is a shared preference-conditioned MLP with two hidden layers of size \(64\). The actor separately encodes the local observation and the preference, then fuses the two embeddings before the action head. The preference embedding dimension is \(32\). The centralized team critic takes the concatenated observations \(\mathbf o\) and outputs a scalar team value, with two hidden layers of size \(128\). Each individual critic has independent parameters, the same input and hidden-layer architecture, and a \(K\)-dimensional output for its agent. Table~\ref{tab:pcma_hyperparams} lists the additional hyperparameters used by PCMA.

\paragraph{Training time.} All PCMA experiments were run on an internal Linux workstation with NVIDIA RTX PRO 6000 Blackwell GPUs (98GB GPU memory each). Each training run used one GPU and one training process. Non-SMAC experiments used \(4\) parallel environment workers, while SMAC used \(1\) environment worker due to the single-environment action-mask path. The measured wall-clock training time for one PCMA seed was \(11.6\) minutes for Cooperative Spread, \(21.2\) minutes for Catch, \(1.12\) hours for Safe Predator-Prey, \(4.49\) hours for MultiWalker, \(1.12\) hours for SMAC-3m, \(2.31\) hours for SMAC-2s3z, and \(1.54\) hours for SMAC-8m. Main comparisons use three seeds, so the reported PCMA training compute is approximately three times these per-seed costs for each environment. Execution-time evaluation and rendering use fixed policy rollouts without gradient updates and were run on the same machine; these jobs are lightweight compared with training and typically finish within minutes per checkpoint.

\begin{algorithm}[htbp]
\caption{Preference Coordinated Multi-Agent Policy Optimization (PCMA)}
\label{alg:pcma_detailed}
\begin{algorithmic}[1]
\Require number of agents \(N\), objective dimension \(K\), rollout length \(T\), PPO epochs \(E\)
\Require actor \(\pi_\theta\), preference planner \(\phi_\psi\), team critic
\(V_\eta^{\mathrm{team}}\), individual critics
\(\{\mathbf V_{\xi_i}^{i}\}_{i=1}^{N}\)
\For{iteration \(=1,2,\ldots\)}
    \For{\(t=0,\ldots,T-1\)}
        \For{agent \(i=1,\ldots,N\)}
            \State Observe local observation \(o_{i,t}\)
            \State Generate inherent preference
            \[
                p_{i,t}\sim \mathrm{Dirichlet}
                \big(\phi_\psi(o_{i,t})\big)
            \]
            \State Sample action
            \[
                a_{i,t}\sim
                \pi_\theta(\cdot\mid o_{i,t},p_{i,t})
            \]
        \EndFor
        \State Execute joint action
        \(\mathbf a_t=(a_{1,t},\ldots,a_{N,t})\)
        \State Store
        \[
        (\mathbf o_t,\{p_{i,t}\}_{i=1}^{N},
        \mathbf a_t,\log\pi_{\theta_{\mathrm{old}}},
        r_t^{\mathrm{team}},
        \{\mathbf r_{i,t}^{\mathrm{ind}}\}_{i=1}^{N},
        \mathbf o_{t+1})
        \]
    \EndFor

    \State Estimate team returns
    \(\hat R^{\mathrm{team}}\) and team advantages
    \(A^{\mathrm{team}}\) with GAE
    \State Estimate individual vector returns
    \(\hat{\mathbf R}_{i}^{\mathrm{ind}}\) and advantages
    \(\mathbf A_{i}^{\mathrm{ind}}\) with GAE
    \State Form each agent's actor advantage:
    \[
        A_i =
        A^{\mathrm{team}}+\lambda_2 p_i^{\top}\mathbf A_i^{\mathrm{ind}}
    \]

    \For{PPO epoch \(=1,\ldots,E\)}
        \For{minibatch \(\mathcal B\)}
            \State Update critics by minimizing
            \[
            \mathcal L_{\mathrm{critic}}
            =
            \mathcal L_{\mathrm{team}}
            +
            \mathcal L_{\mathrm{ind}}.
            \]
            \State Compute PPO ratio
            \[
            \rho_i =
            \frac{
            \pi_\theta(a_i\mid o_i,p_i)
            }{
            \pi_{\theta_{\mathrm{old}}}(a_i\mid o_i,p_i)
            } .
            \]
            \State Update actor with clipped PPO objective:
            \[
            \mathcal L_{\mathrm{actor}}
            =
            -\mathbb E_{\mathcal B}
            \left[
            \min\left(
            \rho_i A_i,
            \mathrm{clip}(\rho_i,1-\epsilon,1+\epsilon)A_i
            \right)
            \right]
            -\beta_{\mathrm{ent}}\mathcal H(\pi_\theta).
            \]
            \State Compute preference diversity regularizer
            \(\mathcal D_\alpha(\{p_i\}_{i=1}^{N})\)
            \State Update planner with
            $
            \mathcal L_{\mathrm{plan}}.
            $
        \EndFor
    \EndFor
\EndFor
\end{algorithmic}
\end{algorithm}

\begin{table}[htbp]
\centering
\caption{Common training setup.}
\label{tab:common_training_setup}
\begin{tabular}{ll}
\toprule
Hyperparameter & Value \\
\midrule
Optimizer & Adam \\
Actor learning rate & \(3\times 10^{-4}\) \\
Critic learning rate & \(3\times 10^{-4}\) \\
Discount factor \(\gamma\) & \(0.99\) \\
GAE parameter \(\lambda\) & \(0.95\) \\
Rollout length & \(256\) \\
PPO epochs \(E\) & \(4\) \\
Minibatches per update & \(4\) \\
PPO clipping coefficient \(\epsilon\) & \(0.2\) \\
Value loss coefficient & \(0.5\) \\
Entropy coefficient & \(0.01\) \\
Max gradient norm & \(0.5\) \\
Actor hidden layers & \([64,64]\) \\
Preference encoder dimension & \(32\) \\
Critic hidden layers & \([128,128]\) \\
Actor parameter sharing & Fully shared across agents \\
Agent-ID embedding & Dimension \(8\) \\
\bottomrule
\end{tabular}
\end{table}

\begin{table}[htbp]
\centering
\caption{PCMA-specific hyperparameters.}
\label{tab:pcma_hyperparams}
\begin{tabular}{lll}
\toprule
Environment & \(\lambda_1\) & \(\lambda_2\) \\
\midrule
Simple Spread & \(0.10\) & \(0.20\) \\
Catch & \(0.05\) & \(0.10\) \\
Escort & \(0.05\) & \(0.10\) \\
Predator-Prey & \(0.05\) & \(0.20\) \\
MultiWalker & \(0.02\) & \(0.20\) \\
SMAC & \(0.02\) & \(0.10\) \\
\bottomrule
\end{tabular}
\end{table}

\newpage
\section{Illustrative Examples}
\subsection{Example for Pareto-Nash Equilibrium}
\label{app:pareto-nash}
 For example, consider a two agent two objective game. Each agent follows a Gaussian policy with fixed variance, $a_1\sim\mathcal{N}(\mu_1,1)$ and $a_2\sim\mathcal{N}(\mu_2,1)$, where $\mu_1,\mu_2\in[-1,1]$. Define team reward as
\[
(r_1, r_2)
=
\left(
\exp\!\left[-\tfrac{1}{2}\Big((a_1-1)^2 + \tfrac{1}{2}(a_2-1)^2\Big)\right],
\;
\exp\!\left[-\tfrac{1}{2}\Big(\tfrac{1}{2}(a_1+1)^2 + 2(a_2+1)^2\Big)\right]
\right).
\]
Under independent action sampling, the expected returns are
\[
\begin{aligned}
R_1(\mu)&=\frac{1}{\sqrt{3}}
\exp\!\left[-\frac{(\mu_1-1)^2}{4}-\frac{(\mu_2-1)^2}{6}\right],
\\
R_2(\mu)&=\frac{\sqrt{2}}{3}
\exp\!\left[-\frac{(\mu_1+1)^2}{6}-\frac{(\mu_2+1)^2}{3}\right].
\end{aligned}
\]
Within the specified fixed-variance Gaussian policy class, $R_1$ is strictly
increasing and $R_2$ is strictly decreasing in either mean on $[-1,1]$.
Consequently, every $(\mu_1,\mu_2)\in[-1,1]^2$ is a Pareto--Nash equilibrium:
no unilateral change can improve one objective without worsening the other.
This does not imply team Pareto optimality. For example, a coordinated change
from $(0,0)$ to $(1/4,-1/4)$ strictly improves both returns, since
\[
\frac{R_1(1/4,-1/4)}{R_1(0,0)}=e^{1/64}>1,
\qquad
\frac{R_2(1/4,-1/4)}{R_2(0,0)}=e^{5/96}>1.
\]

\subsection{A Toy Example of Better Team Improvement with Agent-specific Preferences}
\label{gaussian_example}

We give a one-step, two-agent game with a coupled team objective and two
individual guiding objectives, following the utility definition in
Section~\ref{sec:preference_coordination_objective}. The example uses dense
quadratic rewards to make both the learning dynamics and their limits exact.
Actions are sampled independently as $a_i\sim\mathcal N(\mu_i,1)$, with
$\mu_i\in\mathbb R$ the only trainable parameters. Define
\[
r_{\mathrm{team}}(a_1,a_2)
=-\frac14\left[(a_1-a_2-1)^2+2(a_1+a_2)^2\right],
\qquad
\mathbf r_i(a_i)
=-\frac12\begin{pmatrix}(a_i-1)^2\\(a_i+1)^2\end{pmatrix}.
\]
The team reward favors a prescribed separation and a zero joint center,
whereas the guiding objectives favor actions near $+1$ and $-1$.
Let $s=\mu_1+\mu_2$ and $d=\mu_1-\mu_2$. Taking expectations gives
\[
J_{\mathrm{team}}(\mu)
=-\frac32-\frac14(d-1)^2-\frac12s^2,
\qquad
\mathbf J_i(\mu_i)
=-\frac12\begin{pmatrix}(\mu_i-1)^2+1\\(\mu_i+1)^2+1\end{pmatrix}.
\]
The constant $-3/2$ is the contribution of the fixed action variances.
The unique team optimum in this policy class is
$\mu^{\mathrm{opt}}=(1/2,-1/2)$, with $J_{\mathrm{team}}=-3/2$.

\paragraph{Exact dynamics and equilibria.}
Write $p_i=(\alpha_i,1-\alpha_i)$ and $q_i=2\alpha_i-1$.
Each agent maximizes $U_i=J_{\mathrm{team}}+p_i^\top\mathbf J_i$ using
simultaneous exact gradient ascent with the same step size $\eta$:
\[
\mu^{t+1}=\mu^t+\eta g(\mu^t),
\qquad
g(\mu)=
\begin{pmatrix}
\frac12+q_1-\frac52\mu_1-\frac12\mu_2\\
-\frac12+q_2-\frac12\mu_1-\frac52\mu_2
\end{pmatrix}.
\]
Equivalently,
\[
s^{t+1}=(1-3\eta)s^t+\eta(q_1+q_2),
\qquad
d^{t+1}=(1-2\eta)d^t+\eta(1+q_1-q_2).
\]
Thus, for any fixed preference profile and $0<\eta<2/3$, the iterates
converge from any initial means to the unique stationary point
\[
s^*=\frac{q_1+q_2}{3},
\qquad
d^*=\frac{1+q_1-q_2}{2}.
\]
This point is a Nash equilibrium because each $U_i$ is strictly concave in
its own mean, with $\partial^2 U_i/\partial\mu_i^2=-5/2$.

\paragraph{Comparison with every fixed shared preference.}
For any shared preference $p_1=p_2=(\alpha,1-\alpha)$, let $q=2\alpha-1$.
The induced equilibrium and its team return are
\[
\mu^*_{\mathrm{shared}}=
\left(\frac q3+\frac14,\frac q3-\frac14\right),
\qquad
J_{\mathrm{team}}(\mu^*_{\mathrm{shared}})
=-\frac{25}{16}-\frac{2q^2}{9}
\le-\frac{25}{16}.
\]
The best shared preference is $(1/2,1/2)$.
Instead, choose the interior, agent-specific preferences
\[
p_1=(3/4,1/4),\qquad p_2=(1/4,3/4).
\]
Their average is still $(1/2,1/2)$, but now $q_1=1/2$ and $q_2=-1/2$,
so the induced equilibrium is exactly $\mu^{\mathrm{opt}}=(1/2,-1/2)$.
Its team return is $-3/2$, a strict improvement of $1/16$ over the best
shared-preference equilibrium. Both runs can start from $\mu^0=(0,0)$
and use any common step size $0<\eta<2/3$.

\paragraph{First-order alignment and the role of coordination.}
At $\mu^0=(0,0)$, the team gradient is $(1/2,-1/2)$ and the
team-improvement matrix defined in Section~\ref{sec:first_order_team_improvement} is
\[
B=\begin{pmatrix}1/2&-1/2\\-1/2&1/2\end{pmatrix}.
\]
For the coordinated preferences above, $\bar b=0$,
$\tilde b_i=2\tilde p_i$, and $\mathcal D_p=1/8$.
Hence Assumption~\ref{ass:preference_improvement_alignment} holds at this
initial point with $\kappa=2$, giving an additional first-order term
$\eta N\kappa\mathcal D_p=\eta/2$ relative to the shared mean preference.
Indeed, writing $\Delta J=J_{\mathrm{team}}(\mu^1)-J_{\mathrm{team}}(\mu^0)$,
the exact one-step improvements are
\[
\Delta J_{\mathrm{shared}}=\frac\eta2-\frac{\eta^2}{4},
\qquad
\Delta J_{\mathrm{coordinated}}=\eta-\eta^2.
\]
Their difference, $\eta/2-3\eta^2/4$, is positive for $0<\eta<2/3$;
the quadratic terms also show why the Taylor remainder in
Theorem~\ref{theo:team_improvement_decomposition} must be retained.

Finally, swapping the two coordinated preferences preserves both their mean
and diversity but gives $g(0,0)=0$: learning stays at $(0,0)$ with team return
$-7/4$, worse than the best shared preference. Thus, the benefit comes from
aligning preferences with agents' contributions, not diversity alone.
The example establishes the existence of beneficial preference coordination
under the stated utilities.


\end{document}